\PassOptionsToPackage{hyphens}{url}
\documentclass[conference]{IEEEtran}
\usepackage[hidelinks]{hyperref}
\usepackage{amsmath,amssymb,balance}
\IEEEoverridecommandlockouts

\usepackage{moresize}

\usepackage{amsthm}
\newtheorem{theorem}{Theorem}[section]
\newtheorem{proposition}[theorem]{Proposition}
\newtheorem{lemma}[theorem]{Lemma}
\theoremstyle{definition}
\newtheorem{example}[theorem]{Example}
\newtheorem*{assumption}{Assumption}

\usepackage[binary-units]{siunitx}
\DeclareSIUnit{\txn}{txn}
\DeclareSIUnit{\batch}{batch}
\newcommand{\Replica}[1][r]{\MakeUppercase{#1}}
\newcommand{\Client}[1][c]{\MakeLowercase{#1}}

\newcommand{\m}{\mathbf{m}}
\newcommand{\n}{\mathbf{n}}
\newcommand{\f}{\mathbf{f}}
\newcommand{\nf}{\mathbf{nf}}
\newcommand{\Instance}[1]{\mathcal{I}_{#1}}
\newcommand{\Primary}[1]{\mathcal{P}_{#1}}

\newcommand{\Transaction}{T}
\newcommand{\SignMessage}[2]{\langle#2\rangle_{#1}}
\newcommand{\View}{v}
\newcommand{\rn}{\rho}
\newcommand{\Leader}[1]{\mathcal{L}_{#1}}
\newcommand{\StopOp}{\texttt{stop}}
\newcommand{\Stop}[2]{\StopOp(#1; #2)}
\newcommand{\Message}[2]{\textsc{#1}(#2)}
\newcommand{\Hash}[1]{\operatorname{digest}(#1)}

\newcommand{\Name}[1]{\textnormal{\textsc{#1}}}
\newcommand{\PBFT}{\Name{Pbft}}
\newcommand{\ZZ}{\Name{Zyzzyva}}
\newcommand{\HS}{\Name{HotStuff}}
\newcommand{\SBFT}{\Name{Sbft}}
\newcommand{\RBFT}{\Name{Rbft}}
\newcommand{\PoE}{\Name{PoE}}
\newcommand{\MultiBFT}{\Name{RCC}}
\newcommand{\MirBFT}{\Name{MirBFT}}
\newcommand{\RDB}{\Name{ResilientDB}}
\newcommand{\MultiP}{\Name{RCC-P}}
\newcommand{\MultiZ}{\Name{RCC-Z}}
\newcommand{\MultiS}{\Name{RCC-S}}

\newcommand{\Protocol}{\Name{P}}
\newcommand{\BCA}{\Name{bca}}

\newcommand{\Bandwidth}{\mathit{B}}
\newcommand{\TP}[1]{\mathit{T}_{#1}}
\newcommand{\SizeT}{\mathit{st}}
\newcommand{\SizeM}{\mathit{sm}}
\newcommand{\Max}{\textnormal{max}}
\newcommand{\CMax}{\textnormal{cmax}}
\newcommand{\CPBFT}{\textnormal{c}\PBFT{}}

\newcommand{\abs}[1]{\lvert #1 \rvert}
\newcommand{\dsfrac}[2]{#1 / #2}
\newcommand{\union}{\cup}

\newcommand{\difference}{\setminus}
\newcommand{\Concat}{\oplus}
\renewcommand{\div}{\operatorname{div}}
\newcommand{\lfref}[2]{Line~\ref{#1:#2} of Figure~\ref{#1}}

\newcounter{mycounterf}

\usepackage[noend]{algorithmic}
\newcommand{\GETS}{:=}
\newenvironment{myprotocol}{
    \hrule
    \smallskip
    \footnotesize
    \algsetup{linenosize=\tiny}
    \begin{algorithmic}[1]
        \newenvironment{algenumerate}{\begin{enumerate}}{\end{enumerate}}
        \newcommand{\SPACE}{\item[]}
        
        \newcommand{\TITLE}[2]{\item[] \textbf{\underline{##1}} (##2) \textbf{:}\\[0.5pt]}
        \makeatletter
            \newcommand{\EVENT}[1]{\STATE \textbf{event} ##1 \textbf{do}\begin{ALC@g}}
            \newcommand{\ENDEVENT}{\end{ALC@g}}
        \makeatother
}{
    \end{algorithmic}
    \smallskip
    \hrule
}

\usepackage{tikz,pgfplots,pgfplotstable}
\usetikzlibrary{arrows.meta}
\tikzset{
    >=Stealth,
    plot/.append style={baseline,scale=0.475},
    label/.append style={font=\strut\footnotesize},
    dot/.style={circle,scale=0.35,draw=black,fill=black},
}
\pgfplotscreateplotcyclelist{mycyclelist}{
    thick,black   ,every mark/.append style={solid,fill=\pgfplotsmarklistfill},mark=*\\ 
    thick,red     ,every mark/.append style={solid,fill=\pgfplotsmarklistfill},mark=square*\\
    thick,blue    ,every mark/.append style={solid,fill=\pgfplotsmarklistfill},mark=triangle*\\
    thick,teal    ,every mark/.append style={solid,fill=\pgfplotsmarklistfill},mark=pentagon*\\
    thick,brown   ,every mark/.append style={solid,fill=\pgfplotsmarklistfill},mark=halfsquare*\\ 
    thick,orange  ,every mark/.append style={solid,fill=\pgfplotsmarklistfill},mark=halfcircle*\\
    thick,violet  ,every mark/.append style={solid,fill=\pgfplotsmarklistfill,rotate=180},mark=halfdiamond*\\
}
\pgfplotscreateplotcyclelist{mycyclelistex}{
    black   ,every mark/.append style={solid,fill=\pgfplotsmarklistfill},mark=*\\ 
    red     ,every mark/.append style={solid,fill=\pgfplotsmarklistfill},mark=square*\\
}
\pgfplotsset{
    tick label style={font=\Large},
    legend style={font=\LARGE,cells={anchor=west}},
    title style={font=\Large},
    label style={font=\LARGE},
    width=262.5pt,
    height=185pt,
    every axis/.append style={
        ylabel near ticks,
        xlabel near ticks,
        mark size=2.5pt,
        cycle list name=mycyclelist,
        font=\Large,
        y tick label style={
                        /pgf/number format/precision=1,
                        /pgf/number format/fixed,
                        /pgf/number format/fixed zerofill
                    }
    },
    barstyle/.append style={
        ybar,
        bar width={0.5cm},
        enlarge x limits=0.4,
        enlarge y limits={upper=0.025},
        ymin=0,
        xtick=data
    }
}

\newcommand{\plotTheoryTPut}[2]{\begin{tikzpicture}[plot]
        \begin{axis}[title={($\SI{#1}{\txn\per\text{proposal}}$)},ylabel={Throughput (\si{\txn\per\second})},xlabel={Number of replicas ($\n$)},ymin=0,mark size=0.8pt,legend columns=2]
            \addplot table[x={n},y={max}] {#2};
            \addplot table[x={n},y={pbft}] {#2};
            \addplot table[x={n},y={cmax}] {#2};
            \addplot table[x={n},y={cpbft}] {#2};
            \legend{$\TP{\Max}$,$\TP{\PBFT}$,$\TP{\CMax}$,$\TP{\CPBFT}$};
        \end{axis}
    \end{tikzpicture}}

\pgfplotstableread{
n	nf	max	pbft	cmax	cpbft
4	3	81380.20833333333	62600.16025641025	146484.375	88243.59939759035
7	5	40690.104166666664	31300.080128205125	122070.3125	57042.20210280374
10	7	27126.73611111111	20866.720085470086	113932.29166666666	43485.607506361324
13	9	20345.052083333332	15650.040064102563	109863.28125	35439.768145161295
16	11	16276.041666666666	12520.03205128205	107421.875	30006.11033519553
19	13	13563.368055555555	10433.360042735043	105794.27083333333	26057.70217569787
22	15	11625.744047619048	8942.880036630037	104631.69642857142	23046.629169288863
25	17	10172.526041666666	7825.020032051281	103759.765625	20669.27602091633
28	19	9042.24537037037	6955.573361823362	103081.59722222223	18742.108585858583
31	21	8138.020833333333	6260.016025641025	102539.0625	17147.0004180602
34	23	7398.200757575758	5690.923659673659	102095.17045454547	15804.205952716016
37	25	6781.684027777777	5216.6800213675215	101725.26041666666	14657.81850384246
40	27	6260.016025641025	4815.396942800789	101412.25961538461	13667.420433340247
43	29	5812.872023809524	4471.440018315018	101143.97321428571	12803.0345840868
46	31	5425.347222222223	4173.3440170940175	100911.45833333334	12041.94013524264
49	33	5086.263020833333	3912.5100160256407	100708.0078125	11366.592303893905
52	35	4787.071078431372	3682.3623680241326	100528.49264705883	10763.221910820002
55	37	4521.122685185185	3477.786680911681	100368.92361111112	10220.867984838198
58	39	4283.168859649123	3294.745276653171	100226.15131578948	9730.694302503833
61	41	4069.0104166666665	3130.0080128205127	100097.65625	9285.496869202227
64	43	3875.248015873016	2980.9600122100123	99981.39880952382	8879.342700668189
67	45	3699.100378787879	2845.4618298368296	99875.71022727274	8507.300700789841
70	47	3538.269927536232	2721.7460981047934	99779.21195652173	8165.238294314381
73	49	3390.8420138888887	2608.3400106837607	99690.75520833333	7849.665764435695
76	51	3255.2083333333335	2504.00641025641	99609.375	7557.615705614567
79	53	3130.0080128205127	2407.6984714003943	99534.2548076923	7286.548668205879
82	55	3014.081790123457	2318.5244539411206	99464.69907407407	7034.278576667191
85	57	2906.436011904762	2235.720009157509	99400.11160714286	6798.91324262263
88	59	2806.2140804597702	2158.626215738285	99339.97844827587	6578.806519753369
91	61	2712.6736111111113	2086.6720085470088	99283.85416666667	6372.519522892598
94	63	2625.168010752688	2019.3600082712985	99231.35080645162	6178.788966777809
97	65	2543.1315104166665	1956.2550080128203	99182.12890625	5996.5011430622735
}\dataTwentyRQ

\pgfplotstableread{
n	nf	max	pbft	cmax	cpbft
4	3	81380.20833333333	80177.54515599343	146484.375	141804.8160696999
7	5	40690.104166666664	40088.772577996715	122070.3125	115487.52365184485
10	7	27126.73611111111	26725.84838533114	113932.29166666666	105395.27443724946
13	9	20345.052083333332	20044.386288998357	109863.28125	99423.78393665158
16	11	16276.041666666666	16035.509031198686	107421.875	95147.80779450841
19	13	13563.368055555555	13362.92419266557	105794.27083333333	91755.65553628217
22	15	11625.744047619048	11453.935022284775	104631.69642857142	88896.9383420318
25	17	10172.526041666666	10022.193144499179	103759.765625	86394.4759575354
28	19	9042.24537037037	8908.616128443715	103081.59722222223	84148.24263038549
31	21	8138.020833333333	8017.754515599343	102539.0625	82096.92754203362
34	23	7398.200757575758	7288.867741453949	102095.17045454547	80200.44811826036
37	25	6781.684027777777	6681.462096332785	101725.26041666666	78431.19538679003
40	27	6260.016025641025	6167.503473537956	101412.25961538461	76769.31083677866
43	29	5812.872023809524	5726.967511142388	101143.97321428571	75199.98008497078
46	31	5425.347222222223	5345.169677066228	100911.45833333334	73711.80301923545
49	33	5086.263020833333	5011.096572249589	100708.0078125	72295.77014536971
52	35	4787.071078431372	4716.3261856466725	100528.49264705883	70944.59608119889
55	37	4521.122685185185	4454.308064221857	100368.92361111112	69652.27176343589
58	39	4283.168859649123	4219.870797683865	100226.15131578948	68413.75516436141
61	41	4069.0104166666665	4008.8772577996715	100097.65625	67224.75235057085
64	43	3875.248015873016	3817.9783407615923	99981.39880952382	66081.55902810564
67	45	3699.100378787879	3644.4338707269744	99875.71022727274	64980.94354409417
70	47	3538.269927536232	3485.9802241736274	99779.21195652173	63920.058908726285
73	49	3390.8420138888887	3340.7310481663926	99690.75520833333	62896.375525762356
76	51	3255.2083333333335	3207.101806239737	99609.375	61907.62896208826
79	53	3130.0080128205127	3083.751736768978	99534.2548076923	60951.77881671299
82	55	3014.081790123457	2969.538709481238	99464.69907407407	60026.97590469165
85	57	2906.436011904762	2863.483755571194	99400.11160714286	59131.53575677743
88	59	2806.2140804597702	2764.7429364135664	99339.97844827587	58263.916978460904
91	61	2712.6736111111113	2672.584838533114	99283.85416666667	57422.70339309813
94	63	2625.168010752688	2586.3724243868846	99231.35080645162	56606.58916511786
97	65	2543.1315104166665	2505.5482861247947	99182.12890625	55814.366295019696
}\dataFourHRQ

\pgfplotstableread{
method	measure
ne  551054
e   217345
}\dataTPUTfnEXEC

\pgfplotstableread{
method	measure
ne  0.244
e   0.639
}\dataLATfnEXEC

\pgfplotstableread{
method	measure
ns  156277
ed  21855
cmac    104354
}\dataTPUTfnSIGNS

\pgfplotstableread{
method	measure
ns  0.501
ed  3.644
cmac    0.751
}\dataLATfnSIGNS

\pgfplotstableread{
Nodes	M3	MF	MN	SBFT	PBFT	HS	ZYZ
4	309631	286998	300261	104080	178514	7207	2833
16	269932	273564	261271	98649	121185	7150	3330
32	193850	224788	221133	86746	85657	7000	4756
64	102147	130876	121992	46930	52283	6500	7849
91	55920	62599	61009	33602	35593	7207	7919	
}\dataTputNodes

\pgfplotstableread{
Nodes	M3	MF	MN	SBFT	PBFT	HS	ZYZ
4	0.899	1.022	1.07	3.018	1.672	0.032	4.078
16	1.097	1.088	1.15	3.191	2.495	0.033	3.806
32	1.602	1.399	1.6	3.627	3.512	0.032	3.522
64	3.011	2.452	3.2	6.556	5.754	0.033	3.686
91	5.431	4.901	5.9	9.338	8.458	0.034	3.816
}\dataLatNodes

\pgfplotstableread{
Nodes	M3	MF	MN	SBFT	PBFT	HS	ZYZ
4	309487	287723	302317	158571	194716	9072	230990
16	269941	274089	274034	132600	149677	9072	199895
32	194351	225200	232096	111412	95217	9296	154506
64	105066	131580	130544	60397	55903	9096	80997
91	56020	63281	63087	47687	38455	8952	59203
}\dataTputNodesFF

\pgfplotstableread{
Nodes	M3	MF	MN	SBFT	PBFT	HS	ZYZ
4	0.899	1.022	0.931	2.005	1.56 	0.032	1.308
16	1.079	1.046	1.045	2.632	2.046	0.033	1.499
32	1.578	1.369	1.341	2.943	3.304	0.032	2.498
64	2.913	2.352	2.380	5.119	5.576	0.033	4.925
91	5.346	4.887	5.064	6.738	7.801	0.034	6.543
}\dataLatNodesFF

\pgfplotstableread{
Nodes	M3	MF	MN	SBFT	PBFT	HS	ZYZ
10	26014	27548	27379	10820	13523	800	70
50	116291	121650	122240	47354	52676	3800	2212
100	194351	224788	220133	86746	85657	7207	2264
200	209978	278685	271372	103784	119019	9000	6859
400	220023	294020	290908	103656	156548	12000	11129
}\dataTputBatch

\pgfplotstableread{
Nodes	M3	MF	MN	SBFT	PBFT	HS	ZYZ
10	11.399	10.971	11.843	30.464	21.737	0.060	3.88
50	2.654	2.533	2.920	6.597	5.666	0.040	3.684
100	1.578	1.399	1.9	3.627	3.512	0.032	4.004
200	1.477	1.035	1.663	3.028	2.522	0.030	3.499
400	1.408	0.973	1.687	3.027	1.952	0.031	3.416
}\dataLatBatch

\pgfplotstableread{
Nodes	M3	MF	MN	SBFT	PBFT	HS	ZYZ
4	7106	35239	35139	3658	3968	9140	3829	
16	7106	35239	35139	3622	3455	9065	3626	
32	7039	64678	64560	4108	3166	9113	2684	
64	6397	34409	34356	3793	2744	8988	2490	
91	5976	24848	24820	3689	2489	8819	1948	
}\dataTputNodesChain

\pgfplotstableread{
Nodes	M3	MF	MN	SBFT	PBFT	HS	ZYZ
4	.032	.036	.037	0.018	0.017	0.033	0.018
16	.032	.036	.037	0.019	0.02	0.034	0.02
32	.033	.042	.041	0.016	0.023	0.034	0.028
64	.037	.178	.181	0.017	0.027	0.035	0.032
91	.041	.349	.359	0.018	0.032	0.036	0.041
}\dataLatNodesChain

\newcommand{\systemgraph}[5]{
    \begin{tikzpicture}[plot]
        \begin{axis}[barstyle,width={#5},ylabel={#2},symbolic x coords={#3},xticklabels={#4}]
            \addplot table[x={method},y={measure}] {#1};
        \end{axis}
    \end{tikzpicture}}

\newcommand{\axistput}{Throughput (\si{\txn\per\second})}
\newcommand{\axislat}{Latency (\si{\second})}

\newcommand{\resultgraph}[6]{\begin{tikzpicture}[plot]
        \begin{axis}[xlabel={#3},ylabel={#4},title={#2},xtick={#5},#6]
            \addplot table[x={Nodes},y={MN}] {#1};
            \addplot table[x={Nodes},y={MF}] {#1};
            \addplot table[x={Nodes},y={M3}] {#1};
            \addplot table[x={Nodes},y={PBFT}] {#1};
            \addplot table[x={Nodes},y={ZYZ}] {#1};
            \addplot table[x={Nodes},y={SBFT}] {#1};
            \addplot table[x={Nodes},y={HS}] {#1};
        \end{axis}
    \end{tikzpicture}}

\newcommand{\axisnodes}{Number of replicas ($\n$)}
\newcommand{\axisbatches}{Batch size}
\newcommand{\axisticksnodes}{4,16,32,64,91}
\newcommand{\axisticksbatches}{10,50,100,200,400}

\pgfplotstableread{
t	multibft	mirbft
0	30	30
20	30	30
21	27	0
39	27	0
40	30	27
60	30	27
61	24	0
79	24	0
80	27	24
99	27	24
100	30	27
119	30	27
120	30	30
140	30	30
}\dataExampleMirBFT

\pgfplotstableread{
Nodes	MP	MZ	MS	
4	302317	351335	355297
16	274034	369421	365354
32	232096	355997	356344
64	130544	215995	242222
91	63087	175222	210565
}\dataRCCParadigm

\pgfplotstableread{
Nodes	MP	MZ	MS	
4	0.931	0.939	0.975 
16	1.045	8.397	1.365 
32	1.341	11.964	3.564 
64	2.380	19.226	11.736
91	5.064	62.719	18.499
}\dataRCCParadigmLat

\newcommand{\paradigmgraph}[5]{\begin{tikzpicture}[plot]
        \begin{axis}[xlabel={#2},ylabel={#3},xtick={#4},#5]
            \addplot table[x={Nodes},y={MP}] {#1};
            \addplot table[x={Nodes},y={MZ}] {#1};
            \addplot table[x={Nodes},y={MS}] {#1};
        \end{axis}
    \end{tikzpicture}}

\begin{document}

\title{\MultiBFT{}: Resilient Concurrent Consensus for High-Throughput Secure Transaction Processing\thanks{A brief announcement of this work was presented at the 33rd International Symposium on Distributed Computing (DISC 2019)~\cite{discmbft}.}\thanks{This material is based upon work partially supported by the U.S.\ Department of Energy, Office of Science, Office of Small Business Innovation Research, under Award Number DE-SC0020455.}}

\author{\IEEEauthorblockN{Suyash Gupta}
        \IEEEauthorblockA{} \and
        \IEEEauthorblockN{Jelle Hellings}
        \IEEEauthorblockA{\makebox[0pt]{Moka Blox LLC}\\
			  \makebox[0pt]{Exploratory Systems Lab}\\
			  \makebox[0pt]{Department of Computer Science}\\
                          \makebox[0pt]{University of California, Davis}} \and
        \IEEEauthorblockN{Mohammad Sadoghi}
        \IEEEauthorblockA{}
}

\maketitle

\begin{abstract}
Recently, we saw the emergence of consensus-based database systems that promise resilience against failures, strong data provenance, and federated data management. Typically, these fully-replicated systems are operated on top of a primary-backup consensus protocol, which limits the throughput of these systems to the capabilities of a single replica (the primary). 

To push throughput beyond this single-replica limit, we propose \emph{concurrent consensus}. In concurrent consensus, replicas independently propose transactions, thereby reducing the influence of any single replica on performance. To put this idea in practice, we propose our \MultiBFT{} paradigm that can turn any primary-backup consensus protocol into a \emph{concurrent consensus protocol} by running many consensus instances concurrently. \MultiBFT{} is designed with performance in mind and requires minimal coordination between instances. Furthermore, \MultiBFT{} also promises increased resilience against failures. We put the design of \MultiBFT{} to the test by implementing it in \RDB{}, our high-performance resilient blockchain fabric, and comparing it with state-of-the-art primary-backup consensus protocols. Our experiments show that \MultiBFT{} achieves up to $2.75\times$ higher throughput than other consensus protocols and can be scaled to $91$ replicas.
\end{abstract}


\section{Introduction}

Fueled by the emergence of blockchain technology~\cite{blockchain_dist,bit_pedigree,bftbook}, we see a surge in consensus-based data processing frameworks and database systems~\cite{bftbook,caper,hyperledger,blockchaindb,blockmeetdb,blockplane}. This interest can be easily explained: compared to traditional distributed database systems, consensus-based systems can provide \emph{more resilience during failures}, can provide strong support for \emph{data provenance}, and can enable \emph{federated data processing} in a heterogeneous environment with many independent participants. Consequently, consensus-based systems can prevent disruption of service due to software issues or cyberattacks that compromise part of the system, and can aid in improving data quality of data that is managed by many independent parties, potentially reducing the huge societal costs of cyberattacks and bad data.

At the core of consensus-based systems are \emph{consensus protocols} that enable independent participants (e.g., different companies) to manage a single common database by reliably and continuously replicating a unique sequence of transactions among all participants. By design, these consensus protocols are resilient and can deal with participants that have crashed, are unable to participate due to local network, hardware, or software failures, or are compromised and act malicious~\cite{blockchain_iot,wild}. As such, consensus protocols can be seen as the fault-resilient counterparts of classical two-phase and three-phase commit protocols~\cite{2pc,3pc,easyc}. Most practical systems use consensus protocols that follow the classical primary-backup design of \PBFT{}~\cite{pbftj} in which a single replica, the \emph{primary}, proposes transactions by broadcasting them to all other replicas, after which all replicas \emph{exchange state} to determine whether the primary correctly proposes the same transaction to all replicas and to deal with failure of the primary. Well-known examples of such protocols are \PBFT{}~\cite{pbftj}, \ZZ{}~\cite{zyzzyvaj}, \SBFT{}~\cite{sbft}, \HS{}~\cite{hotstuff}, \PoE{}~\cite{poe}, and \RBFT{}~\cite{rbft}, and fully-optimized implementations of these protocols are able to process up-to tens-of-thousands transactions per second~\cite{icdcs}. 

\subsection{The Limitations of Traditional Consensus}\label{ss:limit}

Unfortunately, a close look at the design of primary-backup consensus protocols reveals that their design \emph{underutilized available network resources}, which prevents the \emph{maximization of transaction throughput}: the throughput of these protocols is determined mainly by the outgoing bandwidth of the primary. To illustrate this, we consider the maximum throughput  by which primaries can replicate transactions. Consider a system with $\n$ replicas of which $\f$ are faulty and the remaining $\nf = \n - \f$ are non-faulty.  The maximum throughput $\TP{\max}$ of any such protocol is determined by the outgoing bandwidth $\Bandwidth$ of the primary, the number of replicas $\n$, and the size of transactions $\SizeT$: $\TP{\Max} = \dsfrac{\Bandwidth}{((\n-1)\SizeT)}$. No practical consensus protocol will be able to achieve this throughput, as dealing with crashes and malicious behavior requires substantial state exchange. Protocols such as \ZZ{}~\cite{zyzzyvaj} can come close, however, by optimizing for the case in which no faults occur, this at the cost of their ability to deal with faults efficiently.

For \PBFT{}, the minimum amount of state exchange consists of two rounds in which \Name{Prepare} and \Name{Commit} messages are exchanged between all replicas (a quadratic amount, see Example~\ref{ex:pbft} in Section~\ref{sec:design}). Assuming that these messages have size $\SizeM$, the maximum throughput of \PBFT{} is $\TP{\PBFT} = \dsfrac{\Bandwidth}{((\n-1)(\SizeT + 3\SizeM))}$. To minimize overhead, typical implementations of \PBFT{} group hundreds of transactions together, assuring that $\SizeT \gg \SizeM$ and, hence,  $\TP{\Max} \approx \TP{\PBFT{}}$.

The above not only shows a maximum on throughput, but also that primary-backup consensus protocols such as \PBFT{} and \ZZ{} severely \emph{underutilize resources} of non-primary replicas: when $\SizeT \gg \SizeM$, the primary sends and receives roughly $(\n-1) \SizeT$ bytes, whereas all other replicas only send and receive roughly $\SizeT$ bytes. The obvious solution would be to use several primaries. Unfortunately, recent protocols such as \HS{}~\cite{hotstuff}, \Name{Spinning}~\cite{spin}, and \Name{Prime}~\cite{prime} that regularly \emph{switch primaries} all require that a switch from a primary happens after all proposals of that primary are processed. Hence, such primary switching does load balance overall resource usage among the replicas, but does not address the underutilization of resources we observe.

\subsection{Our Solution: Towards Resilient Concurrent Consensus}
The only way to push throughput of consensus-based databases and data processing systems beyond the limit $\TP{\Max}$, is by better utilizing available resources. In this paper, we propose to do so via \emph{concurrent consensus}, in which we use many primaries that \emph{concurrently} propose transactions. We also propose \MultiBFT{}, a paradigm for the realization of concurrent consensus. Our contributions are as follows:
\begin{enumerate}
    \item First, in Section~\ref{sec:propose}, we propose \emph{concurrent consensus} and show that concurrent consensus can achieve much higher throughput than primary-backup consensus by effectively utilizing all available system resources.
    \item Then, in Section~\ref{sec:design}, we propose \MultiBFT{}, a paradigm for turning any primary-backup consensus protocol into a concurrent consensus protocol and that is designed for maximizing throughput in all cases, even during malicious activity.
    \item Then, in Section~\ref{sec:improve}, we show that \MultiBFT{} can be utilized to make systems \emph{more resilient}, as it can mitigate the effects of order-based attacks and throttling attacks (which are not prevented by traditional consensus protocols), and can provide better load balancing.
    \item Finally, in Section~\ref{sec:eval},  we put the design \MultiBFT{} to the test by implementing it in \RDB{},\footnote{\RDB{} is open-sourced and available at \url{https://resilientdb.com}.} our high-performance resilient blockchain fabric, and compare \MultiBFT{} with state-of-the-art primary-backup consensus protocols. Our comparison shows that \MultiBFT{} answers the promises of \emph{concurrent consensus}: it achieves up to $2.75\times$ higher throughput than other consensus protocols, has a peak throughput of $\SI{365}{\kilo\txn\per\batch}$ and can be easily scaled to $91$ replicas.
\end{enumerate}

\section{The Promise of Concurrent Consensus}\label{sec:propose}

To deal with the underutilization of resources and the low throughput of primary-backup consensus, we propose \emph{concurrent consensus}. In specific, we design for a system that is optimized for high-throughput scenarios in which a plentitude of transactions are available, and we make every replica a \emph{concurrent primary} that is responsible for proposing and replicating some of these transactions. As we have $\nf$ non-faulty replicas, we can expect to always \emph{concurrently} propose at least $\nf$ transactions if sufficient transactions are available. Such concurrent processing has the potential to drastically improve throughput: in each round, each primary will send out one proposal to all other replicas, and receive $\nf-1$ proposals from other primaries. Hence, the maximum concurrent throughput is $\TP{\CMax} = \dsfrac{\nf\Bandwidth}{((\n -1)\SizeT + (\nf - 1)\SizeT)}$.

In practice, of course, the primaries also need to participate in state exchange to determine the correct operations of all concurrent primaries. If we use \PBFT{}-style state exchange, we end up with a concurrent throughput of $\TP{\CPBFT} = \dsfrac{\nf\Bandwidth}{((\n - 1)(\SizeT + 3\SizeM) + (\nf-1)(\SizeT + 4(\n-1)\SizeM))}$. In Figure~\ref{fig:theory_limit}, we have sketched the maximum throughputs $\TP{\Max}$, $\TP{\PBFT}$, $\TP{\CMax}$, and $\TP{\CPBFT}$. As one can see, concurrent consensus not only promises greatly improved throughput, but also sharply reduces the costs associated with scaling consensus. We remark, however, that these figures provide best-case \emph{upper-bounds}, as they only focus on bandwidth usage. In practice, replicas are also limited by computational power and available memory buffers that puts limits on the number of transactions they can process in parallel and can execute (see Section~\ref{ss:rdb}).

\begin{figure}
    \centering
    \plotTheoryTPut{20}{\dataTwentyRQ}\plotTheoryTPut{400}{\dataFourHRQ}
    \caption{Maximum throughput of replication in a system with $\Bandwidth = \SI{1}{\giga\bit\per\second}$, $\n = 3\f+1$, $\nf = 2\f+1$, $\SizeM = \SI{1}{\kibi\byte}$, and individual transactions are $\SI{512}{\byte}$. On the \emph{left}, each proposal groups $20$ transactions ($\SizeT = \SI{10}{\kibi\byte}$) and on the \emph{right}, each proposal groups $400$ transactions ($\SizeT = \SI{2}{\mebi\byte}$).}\label{fig:theory_limit}
\end{figure}

\section{\MultiBFT{}: Resilient Concurrent Consensus}\label{sec:design}

The idea behind concurrent consensus, as outlined in the previous section, is straightforward:   improve overall throughput by using all available resources via concurrency. Designing and implementing a concurrent consensus system that operates correctly, even during crashes and malicious behavior of some replicas, is challenging, however. In this section, we describe how to design correct \emph{consensus protocols} that deliver on the promises of concurrent consensus. We do so by introducing \MultiBFT{}, a paradigm that can turn any primary-backup consensus protocol into a concurrent consensus protocol. At its basis, \MultiBFT{} makes every replica a primary of a consensus-instance that replicates transactions among all replicas. Furthermore, \MultiBFT{} provides the necessary coordination between these consensus-instances to coordinate execution and deal with faulty primaries. To assure resilience and maximize throughput, we put the following design goals in \MultiBFT{}:
\begin{enumerate}
    \renewcommand{\theenumi}{D\arabic{enumi}}
    \item\label{goal:consensus} \MultiBFT{} provides \emph{consensus} among replicas on the client transactions that are to be executed and the order in which they are executed.
    \item\label{goal:client} Clients can interact with \MultiBFT{} to force execution of their transactions and learn the outcome of execution.
    \item\label{goal:paradigm} \MultiBFT{} is a design paradigm that can be applied to any primary-backup consensus protocol, turning it into a concurrent consensus protocol.
    \item\label{goal:always} In \MultiBFT{}, consensus-instances with non-faulty primaries are \emph{always} able to propose transactions at maximum throughput (with respect to the resources available to any replica), this independent of faulty behavior by any other replica.
    \item\label{goal:recover} In \MultiBFT{}, dealing with faulty primaries does not interfere with the operations of other consensus-instances.
\end{enumerate}
Combined, design goals~\ref{goal:always} and~\ref{goal:recover} imply that instances with non-faulty primaries can propose transactions \emph{wait-free}: transactions are proposed concurrent to any other activities and does not require any coordination with other instances.

\subsection{Background on Primary-Backup Consensus and \PBFT{}}\label{ss:back}
 Before we present \MultiBFT{}, we provide the necessary background and notation for primary-backup consensus. Typical primary-backup consensus protocols operate in \emph{views}. Within each view, a primary can propose client transactions, which will then be executed by all non-faulty replicas. To assure that all non-faulty replicas maintain the same state, transactions are required to be \emph{deterministic}: on identical inputs, execution of a transaction must always produce identical outcomes.  To deal with faulty behavior by the primary or by any other replicas during a view, three complimentary mechanisms are used:

\paragraph*{Byzantine commit}
The primary uses a \emph{Byzantine commit algorithm} \BCA{} to propose a client transaction $\Transaction$ to all replicas. Next, \BCA{} will perform state exchange to determine whether the primary successfully proposed a transaction.  If the primary is non-faulty, then all replicas will receive $\Transaction$ and determine success. If the primary is faulty and more than $\f$ non-faulty replicas do \emph{not} receive a proposal or receive different proposals than the other replicas, then the state exchange step of \BCA{} will detect this failure of the primary.

\paragraph*{Primary replacement}
The replicas use a \emph{view-change algorithm} to replace the primary of the current view $\View$ when this primary is detected to be faulty by non-faulty replicas. This view-change algorithm will collect the state of sufficient replicas in view $v$ to determine a correct starting state for the next view $\View+1$ and assign new primary that will propose client transactions in  view $\View+1$.

\paragraph*{Recovery}
A faulty primary can keep up to $\f$ non-faulty replicas \emph{in the dark} without being detected, as $\f$ faulty replicas can cover for this malicious behavior. Such behavior is not detected and, consequently, does not trigger a view-change. Via a \emph{checkpoint algorithm} the at-most-$\f$ non-faulty replicas that are in the dark will learn the proposed client transactions that are successfully proposed to the remaining at-least-$\nf-\f > \f$ non-fault replicas (that are not in the dark).

\begin{example}\label{ex:pbft}
Next, we illustrate these mechanisms in \PBFT{}. At the core of \PBFT{} is the \emph{preprepare-prepare-commit} Byzantine commit algorithm. This algorithm operates in three phases, which are sketched in Figure~\ref{fig:pbft}.
\begin{figure}[t!]
    \centering
    \begin{tikzpicture}[yscale=0.35,xscale=1.05]
        \draw[thick,draw=black!75] (0.75, 4.5) edge[green!50!black!90] ++(7.25, 0)
                                   (0.75,   0) edge ++(7.25, 0)
                                   (0.75,   1) edge ++(7.25, 0)
                                   (0.75,   2) edge ++(7.25, 0)
                                   (0.75,   3) edge[blue!50!black!90] ++(7.25, 0);

        \draw[thin,draw=black!75] (1,   0) edge ++(0, 4.5)
                                  (2, 0) edge ++(0, 4.5)
                                  (3.5,   0) edge ++(0, 4.5)
                                  (5, 0) edge ++(0, 4.5)
                                  (6.5,   0) edge ++(0, 4.5)
                                  (8,   0) edge ++(0, 4.5);

        \node[left] at (0.8, 0) {$\Replica_3$};
        \node[left] at (0.8, 1) {$\Replica_2$};
        \node[left] at (0.8, 2) {$\Replica_1$};
        \node[left] at (0.8, 3) {$\Replica[p]$};
        \node[left] at (0.8, 4.5) {$\Client$};

        \path[->] (1, 4.5) edge node[above=-4pt,xshift=6pt,label] {$\SignMessage{\Client}{\Transaction}$} (2, 3)
                  (2, 3) edge (3.5, 2)
                         edge (3.5, 1)
                         edge (3.5, 0)
                           
                  (3.5, 2) edge (5, 0)
                           edge (5, 1)
                           edge (5, 3)

                  (3.5, 1) edge (5, 0)
                           edge (5, 2)
                           edge (5, 3)
                         
                  (5, 2) edge (6.5, 0)
                         edge (6.5, 1)
                         edge (6.5, 3)

                  (5, 1) edge (6.5, 0)
                         edge (6.5, 2)
                         edge (6.5, 3)
                
                  (6.5, 0) edge (8, 4.5)
                  (6.5, 1) edge (8, 4.5)
                  (6.5, 2) edge (8, 4.5)
                  (6.5, 3) edge (8, 4.5)
                           ;
        
        \node[dot,red] at (6.5, 0) {};
        \node[dot,red] at (6.5, 1) {};
        \node[dot,red] at (6.5, 2) {};
        \node[dot,red] at (6.5, 3) {};
        
        \node (x) at (6.5, 4.75) {} edge[thick,red] (6.5, 0);
        \node[above=-4pt] at (x) {Execute $\SignMessage{\Client}{\Transaction}$};

        \node[below=-4pt,label] at (2.75, 0) {\Name{PrePrepare}};
        \node[below=-4pt,label] at (4.25, 0) {\Name{Prepare}};
        \node[below=-4pt,label] at (5.75, 0) {\Name{Commit}};
        \node[below=-4pt,label] at (7.25, 0) {\Name{Inform}};
    \end{tikzpicture}
    \caption{A schematic representation of the \emph{preprepare-prepare-commit protocol} of \PBFT{}. First, a client $\Client$ requests transaction $\Transaction$ and the primary $\Replica[p]$ proposes $\Transaction$ to all replicas via a \Name{PrePrepare} message. Next, replicas commit to $\Transaction$ via a two-phase message exchange (\Name{Prepare} and \Name{Commit} messages). Finally, replicas execute the proposal and inform the client.}\label{fig:pbft}
\end{figure}
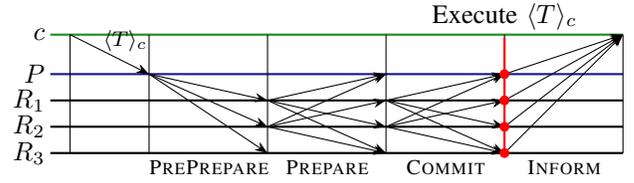

First, the current primary chooses a client request of the form $\SignMessage{\Client}{\Transaction}$, a transaction $\Transaction$ signed by client $\Client$, and proposes this request as the $\rn$-th transaction by broadcasting it to all replicas via a \Name{PrePrepare} message $m$. Next, each non-faulty replica $\Replica$ \emph{prepares} the first proposed $\rn$-th transaction it receives by broadcasting a \Name{Prepare} message for $m$. If a replica $\Replica$ receives $\nf$ \Name{Prepare} messages for $m$ from $\nf$ distinct replicas, then it has the guarantee that any group of $\nf$ replicas will contain a \emph{non-faulty replica} that has received $m$. Hence, $\Replica$ has the guarantee that $m$ can be recovered from any group of $\nf$ replicas, independent of the behavior of the current primary. With this guarantee, $\Replica$ \emph{commits} to $m$ by broadcasting a \Name{Commit} message for $m$. Finally, if a replica $\Replica$ receives $\nf$ \Name{Commit} messages for $m$ from $\nf$ distinct replicas, then it \emph{accepts} $m$. In \PBFT{}, accepted proposals are then \emph{executed} and the client is informed of the outcome.

Each replica $\Replica$ participating in preprepare-prepare-commit uses an internal \emph{timeout value} to detect failure: whenever the primary fails to coordinate a round of preprepare-prepare-commit---which should result in $\Replica$ accepting some proposal---$\Replica$ will detect \emph{failure} of the primary and halt participation in preprepare-prepare-commit. If $\f + 1$ non-faulty replicas detect such a failure and communication is reliable, then they can cooperate to assure that all non-faulty replicas detect the failure. We call this a \emph{confirmed failure} of preprepare-prepare-commit. In \PBFT{}, confirmed failures trigger a \emph{view-change}. Finally, \PBFT{} employs a majority-vote \emph{checkpoint protocol} that allows replicas that are kept in the dark to learn \emph{accepted} proposals without help of the primary.
\end{example}

\subsection{The Design of \MultiBFT{}}

We now present \MultiBFT{} in detail. Consider a primary-backup consensus protocol \Protocol{} that utilizes Byzantine commit algorithm \BCA{} (e.g., \PBFT{} with preprepare-prepare-commit). At the core of applying our \MultiBFT{} paradigm to \Protocol{} is running $\m$, $1 \leq \m \leq \n$, instances of \BCA{} \emph{concurrently}, while providing sufficient coordination between the instances to deal with any malicious behavior.  To do so, \MultiBFT{} makes \BCA{} \emph{concurrent} and uses a checkpoint protocol for per-instance recovery of in-the-dark replicas (see Section~\ref{ss:udetfaulty}). Instead of view-changes, \MultiBFT{} uses a novel wait-free mechanism, that does not involve replacing primaries, to deal with detectable primary failures (see Section~\ref{ss:detfaulty}). \MultiBFT{} requires the following guarantees on \BCA{}:

\begin{assumption}
Consider an instance of \BCA{} running in a system with $\n$ replicas, $\n > 3\f$.
\begin{enumerate}
    \renewcommand{\theenumi}{A\arabic{enumi}}
    \item\label{ass:bca:succeed} If no failures are detected in round $\rn$ of \BCA{} (the round is \emph{successful}), then at least $\nf-\f$ non-faulty replicas have \emph{accepted} a proposed transaction in round $\rn$.
    \item\label{ass:bca:non_diverge} If a non-faulty replica \emph{accepts} a proposed transaction $\Transaction$ in round $\rn$ of \BCA{}, then all other non-faulty replicas that accepted a proposed transaction, accepted $\Transaction$.
    \item\label{ass:bca:recover} If a non-faulty replica \emph{accepts} a transaction $\Transaction$, then $\Transaction$ can be recovered from the state of any subset of $\nf - \f$ non-faulty replicas. 
    \item\label{ass:bca:good} If the primary is non-faulty and communication is reliable, then all non-faulty replicas will accept a proposal in round $\rn$ of \BCA{}.
\end{enumerate}
\end{assumption}
With minor fine-tuning, these assumptions are met by \PBFT{}, \ZZ{}, \SBFT{}, \HS{}, and many other primary-backup consensus protocols, meeting design goal~\ref{goal:paradigm}.

\MultiBFT{} operates in rounds. In each round, \MultiBFT{} replicates $\m$ client transactions (or, as discussed in Section~\ref{ss:limit}, $\m$ sets of client transactions), one for each instance. We write $\Instance{i}$ to denote the $i$-th instance of \BCA{}. To enforce that each instance is coordinated by a distinct primary, the $i$-th replica $\Primary{i}$ is assigned as the primary coordinating $\Instance{i}$. Initially, \MultiBFT{} operates with $\m = \n$ instances. In \MultiBFT{}, instances can fail and be \emph{stopped}, e.g., when coordinated by malicious primaries or during periods of unreliable communication. Each round $\rn$ of \MultiBFT{} operates in three steps:
\begin{enumerate}
\item \emph{Concurrent \BCA{}}. First, each replica participates in $\m$ instances of \BCA{}, in which each instance is proposing a transaction requested by a client among all replicas.
\item \emph{Ordering}. Then, each replica collects all successfully replicated client transactions and puts them in the same---deterministically determined---\emph{order}.
\item \emph{Execution}. Finally, each replica \emph{executes} the transactions of round $\rn$ in order and informs the clients of the outcome of their requested transactions.
\end{enumerate}
Figure~\ref{fig:mbft_replica_overview} sketches a high-level overview of running \emph{$\m$ concurrent instances of \BCA{}}.

\begin{figure}[t!]
    \centering
    \begin{tikzpicture}[xscale=0.7,yscale=0.5]
        \fill[fill=black!10!blue!10] (1, 2) rectangle (3.9, -2.1);
        \fill[fill=black!20!blue!20] (3.9, 2) rectangle (8.75, -2.1);
        \fill[fill=black!30!blue!30] (8.75, 2) rectangle (11.25, -2.1);

        \path (1,  2) edge (1,  -2.1)
              (3.9,  2) edge (3.9,  -2.1)
              (8.75,  2) edge (8.75,  -2.1)
              (11.25, 2) edge (11.25, -2.1);

        \node (i1) at (1.5,  1.5) {$\Instance{1}$};
        \node (i2) at (1.5,  0.5) {$\Instance{2}$};
        \node      at (1.5, -0.25) {$\vdots$};
        \node (im) at (1.5, -1.5) {$\Instance{\m}$};

        \node (bcp1) at (3,  1.5) {\BCA{}};
        \node (bcp2) at (3,  0.5) {\BCA{}};
        \node        at (3, -0.25) {$\vdots$};
        \node (bcpm) at (3, -1.5) {\BCA{}};

        \path[semithick] (3.75, 1.5) edge (3.75, 1.05)
                         (3.75, 1.05) edge (2, 1.05)
                         (2, 1.05) edge[->] (2, 1.5)
                         (3.75, 0.5) edge (3.75, 0.05) 
                         (3.75, 0.05) edge (2, 0.05)
                         (2, 0.05) edge[->] (2, 0.5)
                         (3.75, -1.5) edge (3.75, -1.95) 
                         (3.75, -1.95) edge (2, -1.95)
                         (2, -1.95) edge[->] (2, -1.5);

        \path[->,semithick]
                  (1.75, 1.5) edge (2.5, 1.5) (1.75, 0.5) edge (2.5, 0.5) (1.75, -1.5) edge (2.5, -1.5)
                  (3.5, 1.5) edge (4.375, 1.5) (3.5, 0.5) edge (4.375, 0.5) (3.5, -1.5) edge (4.375, -1.5)
                  (6.25, 0) edge (6.75, 0) (8.25, 0) edge (9.25, 0);
        \path[semithick] (6.25, 1.5) edge (6.25, -1.5) (5.875, 1.5) edge (6.25, 1.5) (5.875, 0.5) edge (6.25, 0.5) (5.875, -1.5) edge (6.25, -1.5);

        \node[align=center] (c1) at (5.125,  1.5) {$\SignMessage{\Client_0}{\Transaction_0}$};
        \node[align=center] (c2) at (5.125,  0.5) {$\SignMessage{\Client_1}{\Transaction_1}$};
        \node[align=center]      at (5.125, -0.25) {$\vdots$};
        \node[align=center] (cm) at (5.125, -1.5) {$\SignMessage{\Client_{\m}}{\Transaction_{\m}}$};

        \node[align=center] (permute) at (7.5,   0) {\strut{}order\\\strut{}requests};
        \node[align=center] (execute) at (10,  0) {\strut{}execute\\\strut{}requests};

        \node[above,font=\bfseries,align=center] at (2.5,  2) {\small \strut{}Concurrent \BCA{}};
        \node[above,font=\bfseries,align=center] at (6.375, 2) {\small \strut{}Ordering};
        \node[above,font=\bfseries,align=center] at (10, 2) {\small \strut{}Execution};
    \end{tikzpicture}
    \caption{A high-level overview of \MultiBFT{} running at replica $\Replica$. Replica $\Replica$ participates in $\m$ concurrent instances of \BCA{} (that run independently and continuously output transactions). The instances yield $\m$ transactions, which are executed in a deterministic order.}\label{fig:mbft_replica_overview}
\end{figure}
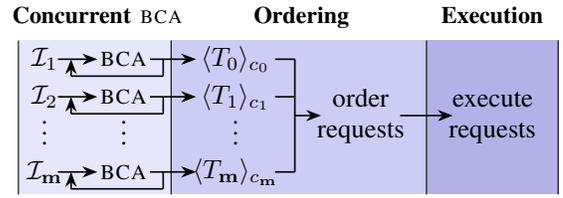

To maximize performance, we want every instance to propose distinct transactions, such that every round results in $\m$ distinct transactions. In Section~\ref{ss:clients}, we delve into the details by which primaries can choose transactions to propose.

To meet design goal~\ref{goal:always} and~\ref{goal:recover}, individual \BCA{} instances in \MultiBFT{} can continuously propose and replicate transactions: ordering and execution of the transactions replicated in a round by the $\m$ instances is done \emph{in parallel} to the proposal and replication of transactions for future rounds. Consequently, non-faulty primaries can utilize their entire outgoing network bandwidth for proposing transactions, even if other replicas or primaries are acting malicious.    

Let $\SignMessage{\Client_i}{\Transaction_i}$ be the transaction $\Transaction_i$ requested by $\Client_i$ and proposed by $\Primary{i}$ in round $\rn$. After all $\m$ instances complete round $\rn$, each replica can collect the set of transactions $S = \{ \SignMessage{\Client_i}{\Transaction_i} \mid 1 \leq i \leq \m \}$. By Assumption~\ref{ass:bca:non_diverge}, all non-faulty replicas will obtain the same set $S$. Next, all replicas choose an order on $S$ and execute all transactions in that order. For now, we assume that the transaction $\SignMessage{\Client_i}{\Transaction_i}$ is executed as the $i$-th transaction of round $\rn$. In Section~\ref{sec:improve}, we show that a more advanced ordering-scheme can further improve the resilience of consensus against malicious behavior. As a direct consequence of Assumption~\ref{ass:bca:good}, we have the following:

\begin{proposition}
Consider \MultiBFT{} running in a system with $\n$ replicas, $\n > 3\f$. If all $\m$ instances have non-faulty primaries and communication is reliable, then, in each round, all non-faulty replicas will accept the same set of $\m$ transactions and execute these transactions in the same order.
\end{proposition}

As all non-faulty replicas will execute each transaction in $\SignMessage{\Client_i}{\Transaction_i} \in S$, there are $\nf$ distinct non-faulty replicas that can inform the client of the outcome of execution. As all non-faulty replicas operate deterministically and execute the transactions in the same order, client $\Client_i$ will receive identical outcomes of $\nf > \f$ replicas, guaranteeing that this outcome is correct.

In the above, we described the normal-case operations of \MultiBFT{}. As in normal primary-backup protocols, individual instances in \MultiBFT{} can be subject to both \emph{detectable} and \emph{undetectable} failures. Next, we deal with these two types of failures.

\subsection{Dealing with Detectable Failures}\label{ss:detfaulty}

Consensus-based systems typically operate in an environment with asynchronous communication: messages can get lost, arrive with arbitrary delays, and in arbitrary order. Consequently, it is impossible to distinguish between, on the one hand, a primary that is malicious and does not send out proposals and, on the other hand, a primary that does send out proposals that get lost in the network. As such, asynchronous consensus protocols can only provide \emph{progress} in periods of \emph{reliable bounded-delay communication} during which all messages sent by non-faulty replicas will arrive at their destination within some maximum delay~\cite{flp,capproof}.

To be able to deal with failures, \MultiBFT{} assumes that \emph{any failure} of non-faulty replicas to receive proposals from a primary $\Primary{i}$, $1 \leq i \leq \m$, is due to \emph{failure} of $\Primary{i}$, and we design the recovery process such that it can also recover from failures due to unreliable communication. Furthermore, in accordance with the wait-free design goals~\ref{goal:always} and~\ref{goal:recover}, the recovery process will be designed so that it does not interfere with other \BCA{} instances or other recovery processes. Now assume that primary $\Primary{i}$ of $\Instance{i}$, $1 \leq i \leq \m$, fails in round $\rn$. The recovery process consists of three steps:
\begin{enumerate}
\item All non-faulty replicas need to detect failure of the $\Primary{i}$.
\item All non-faulty replicas need to reach agreement on the state of $\Instance{i}$: which transactions have been proposed by $\Primary{i}$ and have been accepted in the rounds up-to-$\rn$.
\item To deal with unreliable communication, all non-faulty replicas need to determine the round in which $\Primary{i}$ is allowed to resume its operations.
\end{enumerate}

To reach agreement on the state of $\Instance{i}$, we rely on a separate instance of the consensus protocol $\Protocol$ that is only used to coordinate agreement on the state of $\Instance{i}$ during failure. This coordinating consensus protocol $\Protocol$ replicates $\Stop{i}{E}$ operations, in which $E$ is a set of $\nf$ \Name{Failure} messages sent by $\nf$ distinct replicas from which all accepted proposals in instance $\Instance{i}$ can be derived. We notice that $\Protocol$ is---itself---an instance of a primary-backup protocol that is coordinated by some primary $\Leader{i}$ (based on the current view in which the instance of $\Protocol$ operates), and we use the standard machinery of $\Protocol$ to deal with failures of that leader (see Section~\ref{ss:back}). Next, we shall describe how the recovery process is initiated. The details of this protocol can be found in Figure~\ref{fig:rec_init}.

\begin{figure}[t!]
    \begin{myprotocol}
        \TITLE{Recovery request role}{used by replica $\Replica$}
        \EVENT{$\Replica$ detects failure of the primary $\Primary{i}$, $1 \leq i \leq \m$, in round $\rn$}\label{fig:rec_init:detect}
            \STATE $\Replica$ halts $\Instance{i}$.
            \STATE Let $P$ be the state of $\Replica$ in accordance to Assumption~\ref{ass:bca:recover}.
            \STATE Broadcast $\Message{Failure}{i, \rn, P}$ to all replicas.
        \ENDEVENT
        \EVENT{$\Replica$ receives $\f+1$ messages $m_j = \Message{Failure}{i, \rn_j, P_j}$ such that:
            \begin{algenumerate}
                \item these messages are sent by a set $S$ of $\abs{S} = \f+1$ distinct replicas;
                \item all $\f+1$ messages are well-formed; and
                \item $\rn_j$, $1 \leq j \leq \f+1$, comes after the round in which $\Instance{i}$ started last
            \end{algenumerate}}\label{fig:rec_init:join}
            \STATE $\Replica$ detects failure of $\Primary{i}$ (if not yet done so).
        \ENDEVENT
        \SPACE
        \TITLE{Recovery leader role}{used by leader $\Leader{i}$ of $\Protocol$}
        \EVENT{$\Leader{i}$ receives $\nf$ messages $m_j = \Message{Failure}{i, \rn_j, P_j}$ such that
            \begin{algenumerate}
                \item these messages are sent by a set $S$ of $\abs{S} = \f+1$ distinct replicas;
                \item all $\nf$ messages are well-formed; and
                \item $\rn_j$, $1 \leq j \leq \f+1$, comes after the round in which $\Instance{i}$ started last
            \end{algenumerate}}\label{fig:rec_init:stop}
            \STATE Propose $\Stop{i}{\{ m_1, \dots, m_{\nf} \}}$ via $\Protocol$.
        \ENDEVENT
        \SPACE
        \TITLE{State recovery role}{used by replica $\Replica$}
        \EVENT{$\Replica$ accepts $\Stop{i}{E}$ from $\Leader{i}$ via $\Protocol$}
            \STATE Recover the state of $\Instance{i}$ using $E$ in accordance to Assumption~\ref{ass:bca:recover}.
            \STATE Determine the last round $\rn$ for which $\Instance{i}$ accepted a proposal.
            \STATE Set $\rn + 2^f$, with $f$ the number of accepted $\Stop{i}{E'}$ operations, as the next valid round number for instance $\Instance{i}$.\label{fig:rec_init:restart}
        \ENDEVENT
    \end{myprotocol}
    \caption{The \emph{recovery algorithm} of \MultiBFT{}.}\label{fig:rec_init}
\end{figure}

When a replica $\Replica$ detects failure of instance $\Instance{i}$, $0 \leq i < \m$, in round $\rn$, it broadcasts a message $\Message{Failure}{i, \rn, P}$, in which $P$ is the state of $\Replica$ in accordance to Assumption~\ref{ass:bca:recover} (\lfref{fig:rec_init}{detect}). To deal with unreliable communication, $\Replica$ will continuously broadcast this \Name{Failure} message with an exponentially-growing delay until it learns on how to proceed with $\Instance{i}$. To reduce communication in the normal-case operations of $\Protocol$, one can send the full message $\Message{Failure}{i, \rn, P}$ to only $\Leader{i}$, while sending $\Message{Failure}{i, \rn}$ to all other replicas.

If a replica receives $\f+1$ \Name{Failure} messages from distinct replicas for a certain instance $\Instance{i}$, then it received at least one such message from a non-faulty replica. Hence, it can detect failure of $\Instance{i}$ (\lfref{fig:rec_init}{join}).  Finally, if a replica $\Replica$  receives $\nf$ \Name{Failure} messages from distinct replicas for a certain instance $\Instance{i}$, then we say there is a \emph{confirmed failure}, as $\Replica$ has the guarantee that eventually---within at most two message delays---also the primary $\Leader{i}$ of $\Protocol$ will receive $\nf$ \Name{Failure} messages (if communication is reliable). Hence, at this point, $\Replica$ sets a timer based on some internal timeout value (that estimates the message delay) and waits on the leader $\Leader{i}$ to propose a valid \StopOp{}-operation or for the timer to run out. In the latter case, replica $\Replica$ detects failure of the leader $\Leader{i}$ and follows the steps of a view-change in $\Protocol$ to (try to) replace $\Leader{i}$. When the leader $\Leader{i}$ receives $\nf$ \Name{Failure} messages, it can and must construct a valid \StopOp{}-operation and reach consensus on this operation (\lfref{fig:rec_init}{stop}). After reaching consensus, each replica can recover to a common state of $\Instance{i}$:

\begin{theorem}\label{thm:term}
Consider \MultiBFT{} running in a system with $\n$ replicas. If $\n > 3\f$, an instance $\Instance{i}$, $0 \leq i < \m$, has a confirmed failure, and the last proposal of $\Primary{i}$ accepted by a non-faulty replica was in round $\rn$, then---whenever communication becomes reliable---the recovery protocol of Figure~\ref{fig:rec_init} will assure that all non-faulty replicas will recover the same state, which will include all proposals accepted by non-faulty replicas before-or-at round $\rn$.
\end{theorem}
\begin{proof}
If communication is reliable and instance $\Instance{i}$ has a confirmed failure, then all non-faulty replicas will detect this failure and send \Name{Failure} messages (\lfref{fig:rec_init}{detect}). Hence, all replicas are guaranteed to receive at least $\nf$ \Name{Failure} messages, and any replica will be able to construct a well-formed operation $\Stop{i}{E}$. Hence, \Protocol{} will eventually be forced to reach consensus on $\Stop{i}{E}$. Consequently, all non-faulty replicas will conclude on the same state for instance $\Instance{i}$. Now consider a transaction $\Transaction$ accepted by non-faulty replica $\Replica[q]$ in instance $\Instance{i}$. Due to Assumption~\ref{ass:bca:recover}, $\Replica[q]$ will only accept $\Transaction$ if $\Transaction$ can be recovered from the state of any set of $\nf-\f$ non-faulty replicas. As $\abs{E} = \nf$ (\lfref{fig:rec_init}{stop}), the set $E$ contains the state of $\nf-\f$ non-faulty replicas. Hence, $\Transaction$ must be recoverable from $E$.
\end{proof}

We notice that the recovery algorithm of \MultiBFT{}, as outlined in Figure~\ref{fig:rec_init}, only affects the capabilities of the \BCA{} instance that is stopped. All other \BCA{} instances can concurrently propose transactions for current and for future rounds. Hence, the recovery algorithm adheres to the wait-free design goals~\ref{goal:always} and~\ref{goal:recover}. Furthermore, we reiterate that we have separate instance of the coordinating consensus protocol for each instance $\Instance{i}$, $1\leq i \leq \m$. Hence, recovery of several instances can happen concurrently, which minimizes the time it takes to recover from several simultaneous primary failures and, consequently, minimizes the delay before a round can be executed during primary failures. 

Confirmed failures not only happen due to malicious behavior. Instances can also fail due to periods of unreliable communication. To deal with this, we eventually restart any stopped instances. To prevent instances coordinated by malicious replicas to continuously cause recovery of their instances, every failure will incur an exponentially growing restart penalty (\lfref{fig:rec_init}{restart}). The exact round in which an instance can resume operations can be determined deterministically from the accepted history of \StopOp{}-requests. When all instances have round failures due to unreliable communication (which can be detected from the history of \StopOp{}-requests), any instance is allowed to resume operations in the earliest available round (after which all other instances are also required to resume operations).

\subsection{Dealing with Undetectable Failures}\label{ss:udetfaulty}

As stated in Assumption~\ref{ass:bca:succeed}, a malicious primary $\Primary{i}$ of a \BCA{} instance $\Instance{i}$ is able to keep up to $\f$ non-faulty replicas in the dark without being detected. In normal primary-backup protocols, this is not a huge issue: at least $\nf - \f > \f$ non-faulty replicas still accept transactions, and these replicas can execute and reliably inform the client of the outcome of execution. This is not the case in \MultiBFT{}, however:

\begin{example}\label{ex:collusion}
Consider a system with $\n = 3\f+1 = 7$ replicas. Assume that primaries $\Primary{1}$ and $\Primary{2}$ are malicious, while all other primaries are non-faulty. We partition the non-faulty replicas into three sets $A_1$, $A_2$, and $B$ with $\abs{A_1} = \abs{A_2} = \f$ and $\abs{B} = 1$. In round $\rn$, the malicious primary $\Primary{i}$, $i \in \{1,2\}$, proposes transaction $\SignMessage{\Client_i}{\Transaction_i}$ to only the non-faulty replicas in $A_i \union B$. This situation is sketched in Figure~\ref{fig:bi_attack}. After all concurrent instances of \BCA{} finish round $\rn$, we see that the replicas in $A_1$ have accepted $\SignMessage{\Client_1}{\Transaction_1}$, the replicas in $A_2$ have accepted $\SignMessage{\Client_2}{\Transaction_2}$, and only the replica in $B$ has accepted both $\SignMessage{\Client_1}{\Transaction_1}$ and $\SignMessage{\Client_2}{\Transaction_2}$. Hence, only the single replica in $B$ can proceed with execution of round $\rn$. Notice that, due to Assumption~\ref{ass:bca:succeed}, we consider all instances as finished successfully. If $\n \geq 10$ and $\f \geq 3$, this example attack can be generalized such that also the replica in $B$ is missing at least a single client transaction.
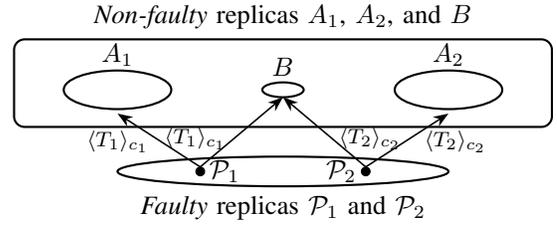
\begin{figure}[t!]
    \centering
    \begin{tikzpicture}[yscale=0.395,xscale=1.1]
        \node[above,align=right] at (0, 1.75) {\emph{Non-faulty} replicas $A_1$, $A_2$, and $B$};
        \draw[thick,rounded corners] (-3.25, -1.25) rectangle (3.25, 1.7);
        \draw[thick]  (0, 0) circle (0.25);
        \node[above=-1pt] (n1) at (0, 0.25) {$B$};
        
        \draw[thick]  (-2, 0) circle (0.65);
        \node[above=-0.25pt] (n2) at (-2, 0.5) {$A_1$};
        
        \draw[thick]  (2, 0) circle (0.65);
        \node[above=-0.25pt] (n3) at (2, 0.5) {$A_2$};
        
        \draw[thick] (0, -2.75) ellipse (2 and 0.5);
        \node[below,align=right] (f) at (0, -3.25) {\emph{Faulty} replicas $\Primary{1}$ and $\Primary{2}$};
        
        \node[dot] (p1) at (-1, -2.75) {};
        \node[dot] (p2) at ( 1, -2.75) {};
        \node[right] at (p1) {$\Primary{1}$};
        \node[left] at (p2) {$\Primary{2}$};
        
        \draw[semithick] (p1.north) edge[->] node[left,label] {$\SignMessage{\Client_1}{\Transaction_1}$} (-2, -0.8) edge[->] node[left,label,yshift=-2pt,xshift=-2pt] {$\SignMessage{\Client_1}{\Transaction_1}$} (0, -0.25)
                         (p2.north) edge[->] node[right=3pt,label] {$\SignMessage{\Client_2}{\Transaction_2}$} ( 2, -0.8) edge[->] node[right,label,yshift=-2pt,xshift=2pt] {$\SignMessage{\Client_2}{\Transaction_2}$} (0, -0.25);
    \end{tikzpicture}
    \caption{An attack possible when parallelizing \BCA{}: malicious primaries can prevent non-faulty replicas from learning all client requests in a round, thereby preventing timely round execution. The faulty primary $\Primary{i}$, $i  \in \{1,2\}$, does so by only letting non-faulty replicas $A_i \union B$ participate in instance $\Instance{i}$.}\label{fig:bi_attack}
\end{figure}
\end{example}

To deal with \emph{in-the-dark} attacks of Example~\ref{ex:collusion}, we can run a standard checkpoint algorithm for each \BCA{} instance: if the system does not reach confirmed failure of $\Primary{i}$ in round $\rn$, $1\leq i \leq \m$, then, by Assumption~\ref{ass:bca:succeed} and~\ref{ass:bca:non_diverge}, at-least-$\nf - \f$ non-faulty replicas have accepted the same transaction $\Transaction$ in round $\rn$ of $\Instance{i}$. Hence, by Assumption~\ref{ass:bca:recover}, a standard checkpoint algorithm (e.g., the one of \PBFT{} or one based on delayed replication~\cite{delay_rep}) that exchanges the state of these at-least-$\nf - \f$ non-faulty replicas among all other replicas is sufficient to assure that all non-faulty replicas eventually accept $\Transaction$. We notice that these checkpoint algorithms can be run concurrently with the operations of \BCA{} instances, thereby adhering to our wait-free design goals~\ref{goal:always} and~\ref{goal:recover}.

To reduce the cost of checkpoints, typical consensus systems only perform checkpoints after every $x$-th round for some system-defined constant $x$. Due to in-the-dark attacks, applying such a strategy to \MultiBFT{} means choosing between execution latency and throughput. Consequently, in \MultiBFT{} we do checkpoints on a dynamic \emph{per-need basis}: when replica $\Replica$ receives $\nf - \f$ claims of failure of primaries (via the \Name{Failure} messages of the recovery protocol) in round $\rn$ and $\Replica$ itself finished round $\rn$ for all its instances, then it will participate in any attempt for a checkpoint for round $\rn$. Hence, if an in-the-dark attack affects more than $\f$ distinct non-faulty replicas in round $\rn$, then a successful checkpoint will be made and all non-faulty replicas recover from the attack, accept all transactions in round $\rn$, and execute all these transactions.

Using Theorem~\ref{thm:term} to deal with detectable failures and using checkpoint protocols to deal with replicas in-the-dark, we conclude that \MultiBFT{} adheres to design goal~\ref{goal:consensus}:

\begin{theorem}
Consider \MultiBFT{} running in a system with $\n$ replicas. If $\n > 3\f$, then \MultiBFT{} provides consensus in periods in which communication is reliable.
\end{theorem}

\subsection{Client Interactions with \MultiBFT{}}\label{ss:clients}

To maximize performance, it is important that every instance proposes distinct client transactions, as proposing the same client transaction several times would reduce throughput. We have designed \MultiBFT{} with faulty clients in mind, hence, we do not expect cooperation of clients to assure that they send their transactions to only a single primary.

To be able to do so, the design of \MultiBFT{} is optimized for the case in which there are always many more concurrent clients than replicas in the system. In this setting, we assign every client $\Client$ to a single primary $\Primary{i}$, $1 \leq i \leq \m = \n$, such that only instance $\Instance{i}$ can propose client requests of $\Client$. For this design to work in \emph{all cases}, we need to solve two issues, however: we need to deal with situations in which primaries do not receive client requests (e.g., during downtime periods in which only few transactions are requested), and we need to deal with faulty primaries that refuse to propose requests of some clients.

First, if there are less concurrent clients than replicas in the system, e.g., when demand for services is low, then \MultiBFT{} still needs to process client transactions correctly, but it can do so without optimally utilizing resources available, as this would not impact throughput in this case due to the low demands. If a primary $\Primary{i}$, $1 \leq i \leq \m$, does not have transactions to propose in any round $\rn$ and $\Primary{i}$ detects that other \BCA{} instances are proposing for round $\rn$ (e.g., as it receives proposals), then $\Primary{i}$ proposes a small no-op-request instead. 

Second, to deal with a primary $\Primary{i}$, $1 \leq i \leq \m$, that refuses to propose requests of some clients, we take a two-step approach. First, we incentivize malicious primaries to \emph{not refuse} services, as otherwise they will be detected faulty and loose the ability to propose transactions altogether. To detect failure of $\Primary{i}$, \MultiBFT{} uses standard techniques to enable a client $\Client$ to \emph{force} execution of a transaction $\Transaction$. First, $\Client$ broadcasts $\SignMessage{\Client}{\Transaction}$ to all replicas. Each non-faulty replica $\Replica$ will then forward $\SignMessage{\Client}{\Transaction}$ to the appropriate primary $\Primary{i}$, $1 \leq i \leq \m$. Next, if the primary $\Primary{i}$ does not propose any transaction requested by $\Client$ within a reasonable amount of time, then $\Replica$ detects failure of $\Primary{i}$. Hence, refusal of $\Primary{i}$ to propose $\SignMessage{\Client}{\Transaction}$ will lead to primary failure, incentivizing malicious primaries to provide service.

Finally, we need to deal with primaries that are unwilling or incapable of proposing requests of $\Client$, e.g., when the primary crashes. To do so, $\Client$ can request to be reassigned to another instance $\Instance{j}$, $1 \leq j \leq \m$, by broadcasting a request $m \GETS \Message{SwitchInstance}{\Client, j}$ to all replicas. Reassignment is handled by the coordinating consensus protocol $\Protocol$ for $\Instance{i}$, that will reach consensus on $m$. Malicious clients can try to use reassignment to propose transactions in several instances at the same time. To deal with this, we assume that no instance is more than $\sigma$ rounds behind any other instance (see Section~\ref{sec:improve}). Now, consider the moment at which replica $\Replica$ accepts $m$ and let $\rn(m, \Replica)$ be the maximum round in which any request has been proposed by any instance in which $\Replica$ participates. The primary $\Primary{i}$ will stop proposing transactions of $\Client$ immediately. Any non-faulty replica $\Replica$ will stop accepting transactions of $\Client$ by $\Instance{i}$ after round $\rn(m, \Replica) + \sigma$ and will start accepting transactions of $\Client$ by $\Instance{j}$ after round $\rn(m, \Replica) + 2\sigma$. Finally, $\Primary{j}$ will start proposing transactions of $\Client$ in round $\rn(m, \Primary{j}) + 3\sigma$.
%Next, we sketch the correctness of this mechanism:
%
%\begin{proof}
%Assume that all instances are proposing rounds within some window $[r, r + \sigma]$). First, notice that $\rho(m, \Replica) \in [r, r + \sigma]$ for all replicas. Now consider replica $\Replica_1$ with $\rho(m, \Replica_1) = r$. This replica will be the first replica to stop accepting requests of $\Client$ by $\Instance{i}$, and will do so after round $r + \sigma$. As $\Instance{i}$ is still proposing for a round in $[r, r + \sigma]$, it might---in the worst case---propose a request of $\Client$ in round $r+\sigma$, which will be accepted by all replicas. Now consider replica $\Replica_2$ with $\rho(m, \Replica_2) = r + \sigma$. This replica will be the last replica to start accepting requests from $\Client$ in instance $\Instance{j}$, and will do so after round $r + \sigma + 2\sigma$.  As $\Primary{j}$ will only start proposing requests of $\Client$ in round $\rho(m, \Primary{j}) + 3\sigma \geq r + 3\sigma$, all replicas will accept such proposals.
%\end{proof}

\section{\MultiBFT{}: Improving Resilience of Consensus}\label{sec:improve}

Traditional primary-backup consensus protocols rely heavily on the operations of their primary. Although these protocols are designed to deal with primaries that \emph{completely fail} proposing client transactions, they are not designed to deal with many other types of malicious behavior. 
\begin{example}
Consider a financial service running on a traditional \PBFT{} consensus-based system. In this setting, a malicious primary can affect operations in two malicious ways:
\begin{enumerate}
\item \emph{Ordering attack}. The primary sets the order in which transactions are processed and, hence, can choose an ordering that best fits its own interests. To illustrate this, we consider client transactions of the form:
\begin{multline*}
    \operatorname{transfer}(A, B, n, m) \GETS\texttt{if $\operatorname{amount}(A) > n$ then}\\
                                             \texttt{$\operatorname{withdraw}(A, m)$; $\operatorname{deposit}(B, m)$.}
\end{multline*}
Let $\Transaction_1 = \operatorname{transfer}(\text{Alice}, \text{Bob}, 500, 200)$ and $\Transaction_2 = \operatorname{transfer}(\text{Bob}, \text{Eve}, 400, 300)$. Before processing these transaction, the balance for Alice is $800$, for Bob $300$, and for Eve $100$. In Figure~\ref{fig:tbl_spendings}, we summarize the results of either first executing $\Transaction_1$ or first executing $\Transaction_2$. As is clear from the figure, execution of $\Transaction_1$  influences the outcome of execution of $\Transaction_2$. As primaries choose the ordering of transactions, a malicious primary can chose an ordering whose outcome benefits its own interests, e.g., formulate targeted attacks to affect the execution of the transaction of some clients.
\begin{figure}[t]
    \centering
    \begin{tabular}{ll||lr||lr}
    \hline
    &Original&\multicolumn{2}{c||}{First $\Transaction_1$, then $\Transaction_2$}&\multicolumn{2}{c}{First $\Transaction_2$, then $\Transaction_1$}\\
        &Balance&$\Transaction_1$ &$\Transaction_2$ &$\Transaction_2$ & $\Transaction_1$\\
    \hline
    \hline
    Alice& 800&600&600&800&600\\
    Bob& 300&500&200&300&500\\
    Eve& 100&100&400&100&100\\
    \hline
    \end{tabular}
    \caption{Illustration of the influence of execution order on the outcome: switching around requests affects the transfer of $\Transaction_2$.}\label{fig:tbl_spendings}
\end{figure}
\item \emph{Throttling attack}. The primary sets the pace at which the system processes transactions. We recall that individual replicas rely on \emph{time-outs} to detect malicious behavior of the primary. This approach will fail to detect or deal with primaries that \emph{throttle throughput} by proposing transactions as slow as possible, while preventing failure detection due to time-outs.
\setcounter{mycounterf}{\value{enumi}}
\end{enumerate}
Besides malicious primaries, also other malicious entities can take advantage of a primary-backup consensus protocol:
\begin{enumerate}
\setcounter{enumi}{\value{mycounterf}}
\item \emph{Targeted attack}. As the throughput of a primary-backup system is entirely determined by the primary, attackers can send arbitrary messages to the primary. Even if the primary recognizes that these messages are irrelevant for its operations, it has spend resources (network bandwidth, computational power, and memory) to do so, thereby reducing throughput. Notice that---in the worst case---this can even lead to failure of a non-faulty primary to propose transactions in a timely manner.
\end{enumerate}
\end{example}

Where traditional consensus-based systems fail to deal with these attacks, the concurrent design of \MultiBFT{} can be used to mitigate these attacks.

First, we look at \emph{ordering attacks}. To mitigate this type of attack, we propose a method to deterministically select a different permutation of the order of execution in every round in such a way that this ordering is practically impossible to predict or influence by faulty replicas. Note that for any sequence $S$ of $k = \abs{S}$ values, there exist $k!$ distinct permutations. We write $P(S)$ to denote these permutations of $S$. To deterministically select one of these permutations, we construct a function that maps an integer $h \in \{0, \dots, k!-1\}$ to a unique permutation in $P(S)$. Then we discuss how replicas will uniformly pick $h$.  As $\abs{P(S)} = k!$, we can construct the following bijection $f_S : \{ 0, \dots, k!-1 \} \rightarrow P(S)$
\[
f_S(i) = 
\begin{cases}
    S &\text{if $\abs{S} = 1$};\\
   f_{S \difference S[q]}(r) \Concat S[q]   &\text{if $\abs{S} > 1$},
\end{cases}
\]
in which $q = i \div {(\abs{S}-1)!}$ is the quotient and $r = i \bmod {(\abs{S}-1)!}$ is the remainder of integer division by $(\abs{S}-1)!$. Using induction on the size of $S$, we can prove:
\begin{lemma}\label{lem:bijection}
$f_S$ is a bijection from $\{ 0, \dots, \abs{S}! - 1 \}$ to all possible permutations of $S$.
\end{lemma}
Let $S$ be the sequence of all transactions accepted in round $\rn$, ordered on increasing instance. The replicas uniformly pick $h = \Hash{S} \bmod (k! - 1)$, in which $\Hash{S}$ is a \emph{strong cryptographic hash function} that maps an arbitrary value $v$ to a numeric digest value in a bounded range such that it is practically impossible to find another value $S'$, $S \neq S'$, with $\Hash{S} = \Hash{S'}$. When at least one primary is non-malicious ($\m > \f$), the final value $h$ is only known after completion of round $\rn$ and it is practically impossible to predictably influence this value. After selecting $h$, all replicas execute the transactions in $S$ in the order given by $f_{S}(h)$.

To deal with primaries that throttle their instances, non-faulty replicas will detect failure of those instances that lag behind other instances. In specific, if an instance $\Instance{i}$, $1 \leq i \leq \m$, is $\sigma$ rounds behind any other instances (for some system-dependent constant $\sigma$), then $\Replica$ detects failure of $\Primary{i}$. 

Finally, we notice that concurrent consensus and \MultiBFT{}---by design---provides \emph{load balancing} with respect to the tasks of the primary, this by spreading the total workload of the system over many primaries. As such, \MultiBFT{} not only improves performance when bounded by the primary bandwidth, but also when performance is bounded by computational power (e.g., due to costly cryptographic primitives), or by message delays. Furthermore, this \emph{load balancing} reduces the load on any single primary to propose and process a given amount of transactions, dampening the effects of any targeted attacks against the resources of a single primary.

\section{Evaluation of the Performance of \MultiBFT{}}\label{sec:eval}
In the previous sections, we proposed concurrent consensus and presented the design of \MultiBFT{}, our concurrent consensus paradigm. To show that concurrent consensus not only provides benefits in theory, we study the performance of \MultiBFT{} and the  effects of concurrent consensus in a practical setting. To do so, we measure the performance of \MultiBFT{} in \RDB{}---our high-performance resilient blockchain fabric---and compare \MultiBFT{} with the well-known primary-backup consensus protocols \PBFT{}, \ZZ{}, \SBFT{}, and \HS{}. With this study, we aim to answer the following questions:
\begin{enumerate}
\renewcommand{\theenumi}{Q\arabic{enumi}}
\item\label{q:1} What is the performance of \MultiBFT{}: does \MultiBFT{} deliver on the promises of concurrent consensus and provide more throughput than any primary-backup consensus protocol can provide?
\item\label{q:2} What is the scalability of \MultiBFT{}: does \MultiBFT{} deliver on the promises of concurrent consensus and provide better scalability than primary-backup consensus protocols?
\item\label{q:3} Does \MultiBFT{} provide sufficient load balancing of primary tasks to improve performance of consensus by offsetting any high costs incurred by the primary? 
\item\label{q:4} How does \MultiBFT{} fare under failures?
\item\label{q:5} What is the impact of batching client transactions on the performance of \MultiBFT{}?
\end{enumerate}
First, in Section~\ref{ss:es}, we describe the experimental setup. Then, in Section~\ref{ss:rdb}, we provide a high-level overview of \RDB{} and of its general performance characteristics. Next, in Section~\ref{ss:cons}, we provide details on the consensus protocols we use in this evaluation. Then, in Section~\ref{ss:meas}, we present the experiments we performed and the measurements obtained. Finally, in Section~\ref{ss:discuss}, we interpret these measurements and answer the above research questions.

\subsection{Experimental Setup}\label{ss:es}
To be able to study the practical performance of \MultiBFT{} and other consensus protocols, we choose to study these protocols in a full resilient database system. To do so, we implemented \MultiBFT{} in \RDB{}. To generate a workload for the protocols, we used the \emph{Yahoo Cloud Serving Benchmark}~\cite{ycsb} provided by the Blockbench macro benchmarks~\cite{blockbench}. In the generated workload, each client transaction queries a YCSB table with half a million active records and $90\%$ of the transactions write and modify records. Prior to the experiments, each replica is initialized with an identical copy of the YCSB table. We perform all experiments in the Google Cloud. In specific, each replica is deployed on a \texttt{c2}-machine with a $16$-core Intel Xeon Cascade Lake CPU running at $\SI{3.8}{\giga\hertz}$ and with $\SI{32}{\giga\byte}$ memory. We use up to $\SI{320}{\kilo{}}$ clients, deployed on $16$ machines.

\subsection{The \RDB{} Blochchain Fabric}\label{ss:rdb}
The \RDB{} fabric incorporates secure permissioned blockchain technologies to provide resilient data processing. A detailed description of how \RDB{} achieves high-throughput consensus in a practical settings can be found in Gupta et al.~\cite{icdcs,resilientdb-demo,tut-middleware19,tut-debs20,tut-vldb20}. The architecture of \RDB{} is optimized for maximizing throughput via \emph{multi-threading} and \emph{pipelining}. To further maximize throughput and minimize the overhead of any consensus protocol, \RDB{} has built-in support for \emph{batching} of client transactions.

We typically group $\SI{100}{\txn\per\batch}$. In this case, the size of a proposal is $\SI{5400}{\byte}$ and of a client reply (for 100 transactions) is $\SI{1748}{\byte}$. The other messages exchanged between replicas during the Byzantine commit algorithm have a size of $\SI{250}{\byte}$. \RDB{} supports \emph{out-of-order processing} of transactions in which primaries can propose future transactions before current transactions are executed. This allows \RDB{} to maximize throughput of any primary-backup protocol that supports out-of-order processing (e.g., \PBFT{}, \ZZ{}, and \SBFT{}) by maximizing bandwidth utilization at the primary.

In \RDB{}, each replica maintains a \emph{blockchain ledger} (a journal) that holds an ordered copy of all executed transactions. The ledger not only stores all transactions, but also proofs of their acceptance by a consensus protocols. As these proofs are built using strong cryptographic primitives, the ledger is \emph{immutable} and, hence, can be used to provide \emph{strong data provenance}. 

In our experiments replicas not only perform consensus, but also communicate with clients and execute transactions. In this practical setting, performance is not fully determined by bandwidth usage due to consensus (as outlined in Section~\ref{ss:limit}), but also by the cost of \emph{communicating with clients}, of sequential \emph{execution} of all transactions, of \emph{cryptography}, and of other steps involved in processing messages and transactions, and by the available memory limitations. To illustrate this, we have measured the effects of \emph{client communication}, \emph{execution}, and \emph{cryptography} on our deployment of \RDB{}.

\begin{figure}
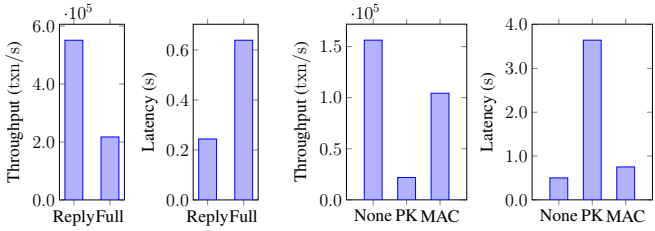

\centering
\setlength{\tabcolsep}{1pt}
\begin{tabular}{c@{\quad}c}
\systemgraph{\dataTPUTfnEXEC}{\axistput}{ne,e}{Reply,Full}{3.4cm}\systemgraph{\dataLATfnEXEC}{\axislat}{ne,e}{Reply,Full}{3.4cm}&
\systemgraph{\dataTPUTfnSIGNS}{\axistput}{ns,ed,cmac}{None,PK,MAC}{4.95cm}\systemgraph{\dataLATfnSIGNS}{\axislat}{ns,ed,cmac}{None,PK,MAC}{4.95cm}
\end{tabular}
\caption{Characteristics of \RDB{} deployed on the Google Cloud. \emph{Left}, the maximum performance of a single replica that receives clients transactions, optionally executes them (Full), and sends replies. \emph{Right}, the performance of \PBFT{} with $\n = 16$ replicas that uses no cryptography (None), uses ED25519 public-key cryptography (PK), or CMAC-AES message authentication codes (MAC) to authenticate messages.}\label{fig:multibft_system}
\end{figure}

In Figure~\ref{fig:multibft_system}, \emph{left}, we present the maximum performance of a single replica that receives clients transactions, optionally executes them (\emph{Full}), and sends replies (without any consensus steps). In this figure, we count the total number of client transactions that are completed during the experiment. As one can see, the system can receive and respond to up-to-$\SI{551}{\kilo\txn\per\sec}$, but can only execute up-to-$\SI{217}{\kilo\txn\per\sec}$.

In Figure~\ref{fig:multibft_system}, \emph{right}, we present the maximum performance of \PBFT{} running on $\n=16$ replicas as a function of the cryptographic primitives used to provide \emph{authenticated communication}. In specific, \PBFT{} can either use digital signatures or message authentication codes. For this comparison,  we compare \PBFT{} using: (1) a baseline that does not use any message authentication (\emph{None}); (2) ED25519 digital signatures for all messages (\emph{DS}); and (3) CMAC+AES message authentication codes for all messages exchanged between messages and ED25519 digital signatures for client transactions. As can be seen from the results, the costs associated with digital signatures are huge, as their usage reduces performance by $86\%$, whereas message authentication codes only reduce performance by $33\%$.

\subsection{The Consensus Protocols}\label{ss:cons}
We evaluate the performance of \MultiBFT{} by comparing it with a representative sample of efficient  practical primary-backup consensus protocols:
\paragraph*{\PBFT{}~\cite{pbftj}} We use a heavily optimized  out-of-order implementation that uses message authentication codes.
\paragraph*{\MultiBFT{}} Our \MultiBFT{} implementation follows the design outlined in this paper. We have chosen to turn \PBFT{} into a concurrent consensus protocol. We test with three variants: \MultiBFT{}$_{\n}$ runs $\n$ concurrent instances, \MultiBFT{}$_{\f+1}$ runs $\f+1$ concurrent instances (the minimum to provide the benefits outlined in Section~\ref{sec:improve}), and \MultiBFT{}$_3$ runs $3$ concurrent instances.
\paragraph*{\ZZ{}~\cite{zyzzyvaj}} As described in Section~\ref{ss:limit}, \ZZ{} has a optimal-case path due to which the performance of \ZZ{} provides an upper-bound for any primary-backup protocol (when no failures occur). Unfortunately, the failure-handling of \ZZ{} is costly, making \ZZ{} unable to deal with any failures efficiently.
\paragraph*{\SBFT{}~\cite{sbft}} This protocol uses threshold signatures to minimize communication during the state exchange that is part of its Byzantine commit algorithm. Threshold signatures do not reduce the communication costs for the primary to propose client transactions, which have a major influence on performance in practice (See  Section~\ref{ss:limit}), but can potentially greatly reduce all other communication costs.
\paragraph*{\HS{}~\cite{hotstuff}} As \SBFT{}, \HS{} uses threshold signatures to minimize communication. The state-exchange of \HS{} has an extra phase compared to \PBFT{}. This additional phase simplifies changing views in \HS{}, and enables \HS{} to regularly switch primaries (which limits the influence of any faulty replicas). Due to this design, \HS{} does not support out-of-order processing (see Section~\ref{ss:limit}). As a consequence, \HS{} is more affected by message delays than by bandwidth. In our implementation, we have used the efficient single-phase event-based variant of \HS{}.

\begin{figure*}
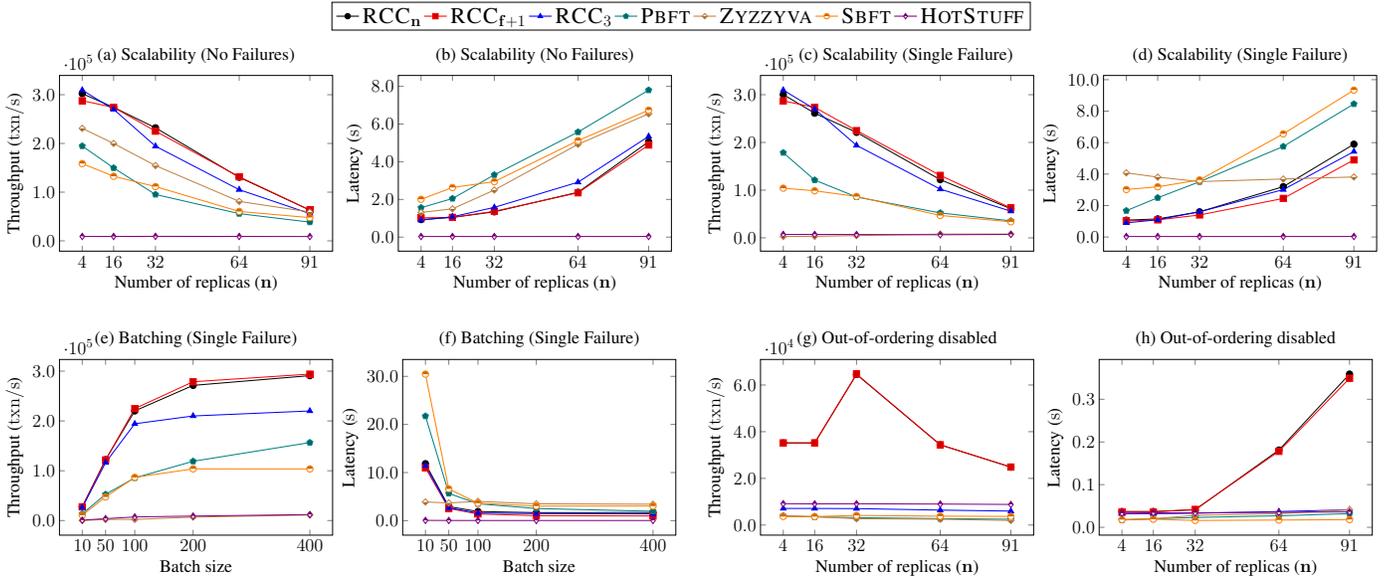

\centering
\scalebox{0.5}{\ref{mainlegend}}\\[5pt]
\setlength{\tabcolsep}{1pt}
\begin{tabular}{cc@{\quad}cc}
\resultgraph{\dataTputNodesFF}{(a) Scalability (No Failures)}{\axisnodes}{\axistput}{\axisticksnodes}{legend to name={mainlegend},legend columns=-1,legend entries={\MultiBFT{}$_{\n}$,\MultiBFT{}$_{\f+1}$,\MultiBFT{}$_3$,\PBFT{},\ZZ{},\SBFT{},\HS{}}}&
\resultgraph{\dataLatNodesFF}{(b) Scalability (No Failures)}{\axisnodes}{\axislat}{\axisticksnodes}{}&
\resultgraph{\dataTputNodes}{~~~~(c) Scalability (Single Failure)}{\axisnodes}{\axistput}{\axisticksnodes}{}&
\resultgraph{\dataLatNodes}{(d) Scalability (Single Failure)}{\axisnodes}{\axislat}{\axisticksnodes}{}\\
\\
\resultgraph{\dataTputBatch}{(e) Batching (Single Failure)}{\axisbatches}{\axistput}{\axisticksbatches}{}&
\resultgraph{\dataLatBatch}{(f) Batching (Single Failure)}{\axisbatches}{\axislat}{\axisticksbatches}{}&
\resultgraph{\dataTputNodesChain}{(g) Out-of-ordering disabled}{\axisnodes}{\axistput}{\axisticksnodes}{}&
\resultgraph{\dataLatNodesChain}{(h) Out-of-ordering disabled}{\axisnodes}{\axislat}{\axisticksnodes}{}
\end{tabular}
\vspace{-3mm}
\caption{Evaluating system throughput and average latency incurred by \MultiBFT{} and other consensus protocols.}\label{fig:multibft_plots}
\vspace{-3mm}
\end{figure*}

\subsection{The Experiments}\label{ss:meas}

To be able to answer Question~\ref{q:1}--\ref{q:5}, we perform four experiments in which we measure the performance of \MultiBFT{}. In each experiment, we measure the \emph{throughput} as the number of transactions that are executed per second, and we measure the \emph{latency} as the time from when a client sends a transaction to the time where that client receives a response.  We run each experiment for $\SI{180}{\second}$: the first $\SI{60}{\second}$ are warm-up, and measurement results are collected over the next $\SI{120}{\second}$. We average our results over three runs. The results of all four experiments can be found in Figure~\ref{fig:multibft_plots}.

In the first experiment, we measure the \emph{best-case performance} of the consensus protocols as a function of the number of replicas when all replicas are non-faulty. We vary the number of replicas between $\n=4$ and $\n=91$ and we use a batch size of $\SI{100}{\txn\per\batch}$. The results can be found in Figure~\ref{fig:multibft_plots}, (a) and (b).

In the second experiment, we measure the performance of the consensus protocols as a function of the number of replicas during failure of a single replica. Again, we vary the number of replicas between $\n=4$ and $\n=91$ and we use a batch size of $\SI{100}{\txn\per\batch}$. The results can be found in Figure~\ref{fig:multibft_plots}, (c) and (d).

In the third experiment, we measure the performance of the consensus protocols as a function of the number of replicas during failure of a single replica while varying the batch size between $\SI{10}{\txn\per\batch}$ and $\SI{400}{\txn\per\batch}$. We use $\n = 32$ replicas. The results can be found in Figure~\ref{fig:multibft_plots}, (e) and (f).

In the fourth and final experiment, we measure the performance of the consensus protocols when outgoing primary bandwidth is not the limiting factor. We do so by disabling \emph{out-of-order processing} in all protocols that support out-of-order processing. This makes the performance of these protocols inherently bounded by the message delay and not by network bandwidth. We study this case by varying the number of replicas between $\n=4$ and $\n=91$ and we use a batch size of $\SI{100}{\txn\per\batch}$. The results can be found in Figure~\ref{fig:multibft_plots}, (g) and (h).

\subsection{Discussion}\label{ss:discuss}
From the experiments, a few obvious patterns emerge. First, we see that increasing the batch size ((e) and (f)) increases performance of all consensus protocols (\ref{q:5}). This is in line with what one can expect (See Section~\ref{ss:limit} and Section~\ref{sec:propose}). As the gains beyond $\SI{100}{\txn\per\batch}$ are small, we have chosen to use $\SI{100}{\txn\per\batch}$ in all other experiments.

Second, we see that the three versions of \MultiBFT{} outperform all other protocols, and the performance of \MultiBFT{} with or without failures is comparable ((a)--(d)). Furthermore, we see that adding concurrency by adding more instances improves performance, as \MultiBFT{}$_{3}$ is outperformed by the other \MultiBFT{} versions. On small deployments with $\n=4, \dots, 16$ replicas, the strength of \MultiBFT{} is most evident, as our \MultiBFT{} implementations approach the maximum rate at which \RDB{} can execute transactions (see Section~\ref{ss:rdb}). 

Third, we see that \MultiBFT{} easily outperforms \ZZ{}, even in the best-case scenario of no failures ((a) and (b)). We also see that \ZZ{} is---indeed---the fastest primary-backup consensus protocol when no failures happen. This underlines the ability of \MultiBFT{}, and of concurrent consensus in general, to reach throughputs no primary-backup consensus protocol can reach.  We also notice that \ZZ{} fails to deal with failures ((c) and (d)), in which case its performance plummets, a case that the other protocols have no issues dealing with. 

Finally, due to the lack of out-of-order processing capabilities in \HS{}, \HS{} is uncompetitive to out-of-order protocols. When we disable out-of-order processing for all other protocols ((g) and (h)), the strength of the simple design of \HS{} shows: its event-based single-phase design outperforms all other primary-backup consensus protocols. Due to the concurrent design of \MultiBFT{}, a non-out-of-order-\MultiBFT{} is still able to greatly outperform \HS{}, however, as the non-out-of-order variants of \MultiBFT{} balance the entire workload over many primaries. Furthermore, as the throughput is not bound by any replica resources in this case (and only by network delays), the non-out-of-order variants \MultiBFT{}$_{\f+1}$ and \MultiBFT$_\n$ benefit from increasing the number of replicas, as this also increases the amount of concurrent processing (due to increasing the number of instances).

\paragraph*{Summary} \MultiBFT{} implementations achieve up to $2.77\times$, $1.53\times$, $38\times$, and $82\times$ higher throughput than 
\SBFT{}, \PBFT{}, \HS{}, and \ZZ{} in single failure experiments.
\MultiBFT{} implementations achieve up to $2\times$, $1.83\times$, $33\times$, and $1.45\times$ higher throughput than 
\SBFT{}, \PBFT{}, \HS{}, and \ZZ{} in no failure experiments, respectively.

Based on these observations, we conclude that \MultiBFT{} delivers on the promises of concurrent consensus. \MultiBFT{} provides more throughput than any primary-backup consensus protocol can provide (\ref{q:1}). Moreover, \MultiBFT{} provides great scalability if throughput is only bounded by the primaries: as the non-out-of-order results show, the load-balancing capabilities of \MultiBFT{} can even offset inefficiencies in other parts of the consensus protocol (\ref{q:2}, \ref{q:3}). Finally, we conclude that \MultiBFT{} can efficiently deal with failures (\ref{q:4}). Hence, \MultiBFT{} meets the design goals~\ref{goal:consensus}--\ref{goal:recover} that we set out in Section~\ref{sec:design}.

\subsection{Analyzing \MultiBFT{} as a Paradigm}
Finally, we experimentally illustrate the ability of \MultiBFT{} to act as a paradigm. To do so, we apply \MultiBFT{} to not only \PBFT{}, but also to \ZZ{} and \SBFT{}. In Figure~\ref{fig:rcc-paradigm}, we plot the performance of these three variants of \MultiBFT{}: \MultiP{} (\MultiBFT{}+\PBFT), \MultiZ{} (\MultiBFT{}+\ZZ), and \MultiS{} (\MultiBFT{}+\SBFT). To evaluate the scalability of these protocols, we perform experiments in the optimistic setting with no failures and $\m= \n$ concurrent instances.

It is evident from these plots that all \MultiBFT{} variants achieve extremely high throughput. As \SBFT{} and \ZZ{} only require linear communication in the optimistic case, \MultiS{} and \MultiZ{} are able achieve up to $3.33\times$ and $2.78\times$ higher throughputs than \MultiP{}, respectively.

Notice that \MultiS{} consistently attains equal or higher throughput than \MultiZ{}, even though \ZZ{} scales better than \SBFT{}. This phenomena is caused by the way \MultiZ{} interacts with clients.  In specific, like \ZZ{}, \MultiZ{} requires its clients to wait for responses of all $\n$ replicas. Hence, clients have to wait longer to place new transactions, and consequently \MultiZ{} requires more clients than \MultiS{} to attain maximum performance. Even if we ran \MultiZ{} with $5$ million clients, the largest amount at our disposal, we would not see maximum performance. Due to the low single-primary performance of \ZZ{}, this phenomena does not prevent \ZZ{} to already reach its maximum performance.

\begin{figure}[t]
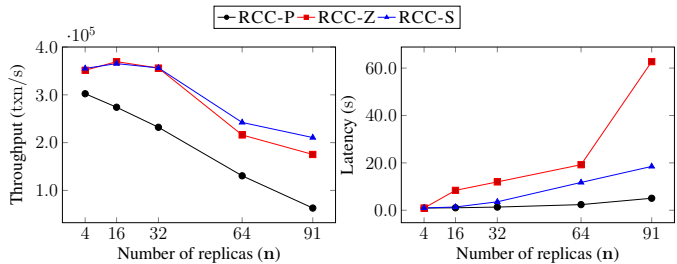

\centering
    \scalebox{0.4}{\ref{paralegend}}\\[5pt]
    \vspace{-2mm}
    \setlength{\tabcolsep}{0.2pt}
    \begin{tabular}{cc}
        \paradigmgraph{\dataRCCParadigm}{\axisnodes}{\axistput}{\axisticksnodes}{legend to name={paralegend},legend columns=-1,legend entries={\MultiP{},\MultiZ{},\MultiS{}}}&
        \paradigmgraph{\dataRCCParadigmLat}{\axisnodes}{\axislat}{\axisticksnodes}{}
    \end{tabular}
    \caption{Evaluating system throughput and latency attained by three \MultiBFT{} variants: \MultiP{}, \MultiZ{} and \MultiS{} when there are no failures.}
    \label{fig:rcc-paradigm}
\end{figure}

\section{Related Work}

In Section~\ref{ss:limit}, we already discussed well-known primary-backup consensus protocols such as \PBFT{}, \ZZ{}, and \HS{} and why these protocols are underutilizing resources. Furthermore, there is abundant literature on consensus and on primary-backup consensus in specific (e.g.,~\cite{wild,scaling,untangle,encybd}). Next, we shall focus on the few works that deal with either \emph{improving throughput and scalability} or with \emph{improving resilience}, the two strengths of \MultiBFT{}

\paragraph*{Parallelizing consensus}
Several recent consensus designs propose to run several primaries concurrently, e.g.,~\cite{rbft,mirbft,omada,sarek}. None of these proposals satisfy all design goals of \MultiBFT{}, however. In specific, these proposals all fall short with respect to maximizing potential throughput in all cases, as none of these proposals satisfy the \emph{wait-free} design goals~\ref{goal:always} and~\ref{goal:recover} of \MultiBFT{}. 

\begin{example}
The \MirBFT{} protocol proposes to run concurrent instances of \PBFT{}, this in a similar fashion as \MultiBFT{}. The key difference is how \MirBFT{} deals with failures: \MirBFT{} operates in global \emph{epochs} in which a \emph{super-primary} decides which instances are enabled. During any failure, \MirBFT{} will switch to a new epoch via a view-change protocol that temporarily shuts-down all instances and subsequently reduces throughput to zero. This is in sharp contrast to the wait-free design of \MultiBFT{}, in which failures are handled on a per-instance level. In Figure~\ref{fig:mirbft}, we illustrated these differences in the failure recovery of \MultiBFT{} and \MirBFT{}.

As is clear from the figure, the fully-coordinated approach of \MirBFT{} results in substantial performance degradation during failure recovery.  Hence, \MirBFT{} does not meet design goals~\ref{goal:always} and~\ref{goal:recover}, which is sharply limits the throughput of \MirBFT{} when compared to \MultiBFT{}.
\end{example}

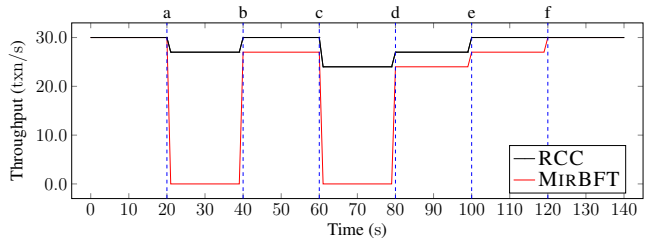
\begin{figure}
    \begin{tikzpicture}[plot]
        \begin{axis}[ylabel={Throughput (\si{\txn\per\second})},xlabel={Time (\si{\second})},mark size=0pt,width=500pt,legend pos=south east,cycle list name=mycyclelistex,xmin=-5,xmax=145,clip=false]
            \addplot table[x={t},y={multibft}] {\dataExampleMirBFT};
            \addplot table[x={t},y={mirbft}] {\dataExampleMirBFT};  
            \addplot table[x={t},y={multibft}] {\dataExampleMirBFT}; 

            \draw[dashed,blue,thick] ({axis cs:20,0}|-{rel axis cs:0,0}) -- ({axis cs:20,0}|-{rel axis cs:0,1})
                                     ({axis cs:40,0}|-{rel axis cs:0,0}) -- ({axis cs:40,0}|-{rel axis cs:0,1})
                                     ({axis cs:60,0}|-{rel axis cs:0,0}) -- ({axis cs:60,0}|-{rel axis cs:0,1})
                                     ({axis cs:80,0}|-{rel axis cs:0,0}) -- ({axis cs:80,0}|-{rel axis cs:0,1})
                                     ({axis cs:100,0}|-{rel axis cs:0,0}) -- ({axis cs:100,0}|-{rel axis cs:0,1})
                                     ({axis cs:120,0}|-{rel axis cs:0,0}) -- ({axis cs:120,0}|-{rel axis cs:0,1});
                                     
            \node[above] at ({axis cs:20,0}|-{rel axis cs:0,1}) {a};
            \node[above] at ({axis cs:40,0}|-{rel axis cs:0,1}) {b};
            \node[above] at ({axis cs:60,0}|-{rel axis cs:0,1}) {c};
            \node[above] at ({axis cs:80,0}|-{rel axis cs:0,1}) {d};
            \node[above] at ({axis cs:100,0}|-{rel axis cs:0,1}) {e};
            \node[above] at ({axis cs:120,0}|-{rel axis cs:0,1}) {f};
        \legend{\MultiBFT{},\MirBFT{}};
        \end{axis}
    \end{tikzpicture}
    \caption{Throughput of \MultiBFT{} versus \MirBFT{} during instance failures with $\m = 11$ instances. At (a), primary $\Primary{1}$ fails. In \MultiBFT{}, all other instances are unaffected, whereas in \MirBFT{} all replicas need to coordinate recovery. At (b), recovery is finished. In \MultiBFT{}, all instances can resume work, whereas \MirBFT{} halts an instance due to recovery. At (c) primaries $\Primary{1}$ and $\Primary{2}$ fail. In \MultiBFT{}, $\Primary{2}$ will be recovered at (d) and $\Primary{1}$ at (e) (as $\Primary{1}$ failed twice, its recovery in \MultiBFT{} takes twice as long). In \MirBFT{}, recovery is finished at (d), after which \MirBFT{} operates with only $\m = 9$ instances. At (e) and (f), \MirBFT{} decides that the system is sufficiently reliable, and \MirBFT{} enables the remaining instances one at a time.}\label{fig:mirbft}
\end{figure}

\paragraph*{Reducing malicious behavior}
Several works have observed that traditional consensus protocols only address a narrow set of malicious behavior, namely behavior that prevents any progress~\cite{rbft,spin,prime,aardvark}. Hence, several designs have been proposed to also address behavior that impedes performance without completely preventing progress. One such design is \RBFT{}, which uses concurrent primaries \emph{not} to improve performance---as we propose---but only to mitigate throttling attacks in a way similar to what we described in Section~\ref{sec:improve}. In practice, the design of \RBFT{} results in poor performance at high costs.

\HS{}~\cite{hotstuff}, \Name{Spinning}~\cite{spin}, and \Name{Prime}~\cite{prime} all proposes to minimize the influence of malicious primaries by replacing the primary every round. This would not incur the costs of \RBFT{}, while still reducing---but not eliminating---the impact of faulty replicas to severely reduce throughput. Unfortunately, these protocols follow the design of primary-backup consensus protocols and, as discussed in Section~\ref{sec:propose}, these designs are unable to achieve throughputs close to those reached by a concurrent consensus such as \MultiBFT{}.

\paragraph*{Concurrent consensus via sharding}
Several recent works have proposed to speed up consensus-based systems by incorporating sharding, this either at the data level (e.g.,~\cite{caper,blockchaindb,cerberus,sharper,ahl}) or at the consensus level (e.g.,~\cite{geobft}). In these approaches only a small subset of all replicas, those in a single shard, participate in the consensus on any given transaction, thereby reducing the costs to replicate this transaction and enabling concurrent transaction processing in independent shards. As such, sharded designs can promise huge scalability benefits for easily-sharded workloads. To do so, sharded designs utilize a weaker failure model than the fully-replicated model \MultiBFT{} uses, however. Consider, e.g., a sharded system with $z$ shards of $\n = 3\f+1$ replicas each. In this setting, the system can only tolerate failure of up to $\f$ replicas in a single shard, whereas a fully-replicated system using $z$ replicas could tolerate the failure of any choice of $\lfloor\dsfrac{(z\n-1)}{3}\rfloor$ replicas. Furthermore, sharded designs typically operate consensus protocols such as \PBFT{} in each shard to order local transactions, which opens the opportunity of concurrent consensus and \MultiBFT{} to achieve even higher performance in these designs.

\section{Conclusion}
In this paper, we proposed \emph{concurrent consensus} as a major step toward enabling high-throughput and more scalable consensus-based database systems. We have shown that concurrent consensus is in theory  able to achieve throughputs that primary-backup consensus systems are unable to achieve. To put the idea of concurrent consensus in practice, we proposed the \MultiBFT{} paradigm that can be used to make normal primary-backup consensus protocols \emph{concurrent}.  Furthermore, we showed that \MultiBFT{} is capable of making consensus-based systems more resilient to failures by sharply reducing the impact of faulty replicas on the throughput and operations of the system. We have also put the design of the \MultiBFT{} paradigm to the test by implementing it in \RDB{}, our high-performance resilient blockchain fabric, and comparing it with state-of-the-art primary-backup consensus protocols. Our experiments show that \MultiBFT{} is able to fulfill the promises of concurrent consensus, as it significantly outperforms other consensus protocols and provides better scalability. As such, we believe that \MultiBFT{} opens the door to the development of new high-throughput resilient database and federated transaction processing  systems.

\paragraph*{Acknowledgements} 
We would like to acknowledge Sajjad Rahnama and Patrick J.\ Liao for their help during the initial stages of this work.

\balance
\bibliographystyle{IEEEtran}
\bibliography{bibliography}

\end{document}